\documentclass[10pt]{article}

\usepackage[authoryear,round]{natbib}                   
\usepackage{amsmath}                                                        
\numberwithin{equation}{section}
\usepackage{graphicx}                                                       
\usepackage{subfigure}
\usepackage{enumitem}                                                    
\usepackage{hyperref}
\usepackage{amssymb}                                                        
\usepackage[mathscr]{eucal}                                             
\usepackage{cancel}                                                             
\usepackage[normalem]{ulem}                                                                 
\usepackage{pstricks}
\usepackage{rotating}
\usepackage{lscape}
\usepackage[paperwidth=8.5in,paperheight=11in,top=1.00in, bottom=1.00in, left=1.00in, right=1.00in]{geometry}
\usepackage{mathtools}                                                      
\mathtoolsset{showonlyrefs=true}                                    
\usepackage{fixltx2e,amsmath}                                           
\MakeRobust{\eqref}
\usepackage{mathdots}
\usepackage{amsthm}                                                             
\allowdisplaybreaks                                                             
\theoremstyle{plain}
\newtheorem{theorem}{Theorem}
\numberwithin{theorem}{section}

\newtheorem{lemma}[theorem]{Lemma}                              
\newtheorem{proposition}[theorem]{Proposition}
\newtheorem{corollary}[theorem]{Corollary}
\theoremstyle{definition}
\newtheorem{definition}[theorem]{Definition}
\newtheorem{example}[theorem]{Example}
\newtheorem{remark}[theorem]{Remark}
\newtheorem{assumption}[theorem]{Assumption}

\def \s {{\sigma}}
\def \g {{\gamma}}

\def \b {{\beta}}
\def \xbar {\bar{x}}
\def \zbar {\bar{z}}

\def \R {\mathbb{R}}
\def \p {\partial}
\def \t {\tau}

\newcommand{\<}{\langle}
\renewcommand{\>}{\rangle}
\renewcommand{\(}{\left(}
\renewcommand{\)}{\right)}
\renewcommand{\[}{\left[}
\renewcommand{\]}{\right]}

\newcommand\Eb{\mathbb{E}}
\newcommand\Pb{\mathbb{P}}

\newcommand\Rb{\mathbb{R}}

\newcommand\Ac{\mathscr{A}}

\newcommand\Bc{\mathscr{B}}

\newcommand\Fc{\mathscr{F}}
\newcommand\Gc{\mathscr{G}}

\newcommand\Lc{\mathscr{L}}
\newcommand\Mc{\mathscr{M}}
\newcommand\Nc{\mathscr{N}}
\newcommand\Oc{\mathscr{O}}
\newcommand\Pc{\mathscr{P}}
\newcommand\Sc{\mathscr{S}}

\newcommand\Jc{\mathscr{J}}

\newcommand\Om{\Omega}
\newcommand\sig{\sigma}

\newcommand\gam{\gamma}
\newcommand\Gam{\Gamma}
\newcommand\lam{\lambda}
\newcommand\del{\delta}

\newcommand\xb{\bar{x}}
\newcommand\yb{\bar{y}}
\newcommand\zb{\bar{z}}
\newcommand\rhob{\overline{\rho}}

\newcommand\Hv{\mathbf{H}}
\newcommand\Cv{\mathbf{C}}
\newcommand\mv{\mathbf{m}}

\newcommand\Xt{\widetilde{X}}
\newcommand\xt{\widetilde{x}}
\newcommand\ut{\widetilde{u}}

\renewcommand\d{\partial}

\newcommand\ii{\mathtt{i}}
\newcommand\dd{\mathrm{d}}
\newcommand\ee{\mathrm{e}}
\newcommand\BS{\mathrm{BS}}

\begin{document}

\title{Explicit implied volatilities for multifactor local-stochastic volatility models}

\author{
Matthew Lorig
\thanks{Department of Applied Mathematics, University of Washington, Seattle, USA.
\textbf{e-mail}: mattlorig@gmail.com.}
\and
Stefano Pagliarani
\thanks{CMAP, Ecole Polytechnique Route de Saclay, 91128 Palaiseau Cedex, France.
\textbf{e-mail}: stepagliara1@gmail.com. Work partially supported by the Chair {\it Financial Risks} of the {\it Risk Foundation}.}
\and
Andrea Pascucci
\thanks{Dipartimento di Matematica, Universit\`a di Bologna, Bologna, Italy.
\textbf{e-mail}: andrea.pascucci@unibo.it}
}

\date{This version: \today}

\maketitle

\begin{abstract}
We consider an asset whose risk-neutral dynamics are described by a general class of local-stochastic
volatility models and derive a family of asymptotic expansions for European-style option prices and implied volatilities.
{We also establish rigorous error estimates for these quantities.}
Our implied volatility expansions are explicit; they do not require any special functions nor do they require numerical integration.  To illustrate the accuracy and versatility of our method, we implement it under
four
different model dynamics: CEV local volatility,
Heston stochastic volatility, $3/2$ stochastic volatility, and SABR local-stochastic volatility.
\end{abstract}

\noindent
\textbf{Keywords}:  implied volatility, local-stochastic volatility, CEV, Heston, SABR.

%
%

\section{Introduction}
\label{sec:intro}
Local-stochastic volatility (LSV) models combine the features of local volatility (LV) and
stochastic volatility (SV) models by describing the instantaneous volatility of an underlying $S$
by a function $\sig(t,S_t,Y_t)$ where $Y$ is some auxiliary, possibly multidimensional, stochastic
process (see, for instance, \cite{lipton2002}, \cite{AlexanderNogueira2004}, \cite{Ewald2005},
\cite{Henry-Labordere2009} and \cite{Clark2010}). Typically, unobservable LSV (or SV or LV) model
parameters are obtained by calibrating to implied volatilities that are observed on the market.
For this reason closed-form approximations for model-induced implied volatilities are useful.  A
number of different approaches have been taken for computing approximate implied volatilities in LV,
SV and LSV models.  We review some of these approaches below.


Concerning LV models, perhaps the earliest and most well-known implied volatility result is due to
\cite{hagan-woodward}, who use singular perturbation methods to obtain an implied volatility
expansion for general LV models.  For certain models (e.g., CEV) they obtain closed-form
approximations. More recently, \cite{lorigCEV} uses regular perturbation methods to obtain an
implied volatility expansion when a LV model can be written as a regular perturbation around
Black-Scholes. \cite{lorig-jacquier} extend and refine the results of \cite{lorigCEV} to find
closed-form approximations of implied volatility for local L\'evy-type models with jumps.
\cite{gatherallocal} examines the small-time asymptotics of implied volatility for LV models
using heat kernel methods.


There is no shortage of implied volatility results for SV models either.
\cite{lorigfouquesircar} (see also \cite{fpss}) derive an asymptotic expansion for general
multiscale stochastic volatility models using combined singular and regular perturbation theory.
\cite{forde2011small} use the Freidlin-Wentzell theory of large deviations for SDEs to obtain
the small-time behavior of implied volatility for general stochastic volatility models with zero
correlation. Their work adds mathematical rigor to previous work by
\cite{lewis-geometry}. \cite{forde2009small} use large deviation techniques to obtain the
small-time behavior of implied volatility in the Heston model (with correlation).  They further
refine these results in \cite{forde-jacquier-lee}.
{More recently, \cite{lorig-jacquier-2} provide an explicit implied volatility approximation for any model with an analytically tractable characteristic function, which includes both affine stochastic volatility and exponential L\'evy models.}


Concerning LSV models, perhaps the most well-known implied volatility result is due to
\cite{sabr}, who use WKB approximation methods to obtain implied volatility asymptotics in a LSV
model with a CEV-like factor of local volatility and a GBM-like auxiliary factor of volatility
(i.e., the SABR model). Another important contribution is due to \cite{berestycki-busca-florent},
who show that implied volatility in an LSV setting can be obtained by solving a quasi-linear
parabolic partial differential equation. More recently, \cite{henry2005general} uses a heat kernel
expansion on a Riemann manifold to derive first order asymptotics for implied volatility for any
LSV model.  As an example, he introduces the $\lambda$-SABR model, which is a LSV model with a
mean reverting auxiliary factor of volatility, and obtains closed form asymptotic formulas for
implied volatility in this setting.  See also \cite{laborderebook}. Finally, we mention
\cite{Watanabe87} and the recent work of \cite{BenhamouGobetMiri2010} and \cite{BompisGobet2012}
who use Malliavin calculus techniques to derive closed-form approximations for implied volatility
in an LSV setting.
There are also some model-free results concerning the extreme-strike behavior of implied
volatility.  Most notably, we mention the work of \cite{lee2004moment} and
\cite{lee2011asymptotics}.

In this paper, we provide closed-form approximations for implied volatility for a very general
class of LSV models. We show (through a series of numerical experiments) that our 
approximation performs favorably when compared to other well-known implied volatility
approximations (e.g., \cite{hagan-woodward} for CEV, \cite{forde-jacquier-lee} for Heston, and
\cite{sabr} for SABR).
Additionally, we prove that our implied volatility expansion satisfies some rigorous error bounds for short-maturities.  As a byproduct of the
implied volatility analysis, we obtain some 
results concerning the short-maturity behavior of the Black-Scholes price, which appear to be new
and of some independent interest.
All of our results are consistent with the
previously derived short-maturity asymptotic results appearing in
{\cite{berestycki2002asymptotics}, \cite{berestycki-busca-florent} and
\cite{BompisGobet2012}.}
The methodology presented in this paper builds upon the asymptotic pricing formulas {first derived in
\cite{pagliarani2011analytical} for scalar diffusions and later extended in \cite{pascucci} and
\cite{lorig-pagliarani-pascucci-1} for scalar L\'evy-type processes.}
\par
The rest of this paper proceeds as follows:
In Section \ref{sec:model}, we introduce a general class of local-stochastic volatility models.  We also derive a family of asymptotic expansions for European option prices and, under certain assumptions, provide rigorous error bounds for our pricing approximation.  In Section \ref{sec:impvol} we translate our asymptotic price expansion into an asymptotic expansion of implied volatility.
{In Section \ref{sec:error} we establish rigorous error estimates for both our pricing and implied volatility expansions.}
Finally, in Section \ref{sec:examples} we test our implied volatility approximation on
four
well-known models: CEV local volatility,
Heston stochastic volatility, three-halves stochastic volatility and SABR local-stochastic
volatility. {Appendix \ref{sec:lemmas} contains the results for the Black-Scholes price at short
maturities.
}

%
%

\section{Asymptotic pricing for a general class of LSV models}
\label{sec:model} For simplicity, we assume a frictionless market, no arbitrage, zero interest
rates and no dividends.  We take, as given, an equivalent martingale measure $\Pb$, chosen by the
market on a complete filtered probability space $(\Om,\Fc,\{\Fc_t,t\geq0\},\Pb)$.  The filtration
$\{\Fc_t,t\geq0\}$ represents the history of the market.  All stochastic processes defined below
live on this probability space and all expectations are taken with respect to $\Pb$.  We consider
a strictly positive asset $S$ whose risk-neutral dynamics are given by $S  =  \exp (X)$ where $X=Z^{(1)}$ is the first component of
a $d$-dimensional diffusion $Z=(X,Y)$, which solves
\begin{align}
\left.
\begin{aligned}
 \dd Z^{(i)}_t
    &=  \mu_{i}(t,Z_t) \dd t + \sig_{i}(t,Z_t) \dd W^{(i)}_t , & Z_0 &= z \in \Rb^{d} , \\
\dd \< W^{(i)}, W^{(j)} \>_t
    &=  \rho_{ij}(t,Z_t) \, \dd t , &
| \rho_{ij} |  &< 1 .
\end{aligned} \right\}
\label{eq:StochVol}
\end{align}
We assume that SDE \eqref{eq:StochVol} has a unique strong solution. 
Sufficient conditions for the existence of a unique strong solution can be found, for example, in
\cite{Ikedabook} or \cite{pascuccibook}.  We also assume that the coefficients are such that $\Eb
[ S_t ]< \infty$ for all $t \in [0,{T_{0}}]$ for some positive $T_{0}$.

Let $V_t$ be the time $t$ value of a European derivative, expiring at time
$T>t$ with payoff $\varphi(X_T)$. Using risk-neutral pricing, 
to value a European-style option we must compute functions of the form
\begin{align}
u(t,x,y)
    &:= \Eb [\varphi(X_T) | X_t = x , Y_t = y] . \label{eq:v}
\end{align}
It is well-known that, under mild assumptions, the function $u$ satisfies the Kolmogorov backward equation
\begin{align}
(\d_t + \Ac(t)) u(t,x,y)
    &=  0 , &
u(T,x,y)
    &=  \varphi(x), \label{eq:v.pde}
\end{align}
where the operator $\Ac(t)$ is given explicitly by
\begin{align}
\Ac(t)
    &=  \frac{1}{2}\sum_{i,j=1}^{d}\rho_{ij}(t,z)\sig_{i}(t,z)\sig_{j}(t,z)\d_{z_{i}z_{j}} + \sum_{i=1}^{d} \mu_{i}(t,z) \d_{z_i}. \label{eq:A}
\end{align}
As a standing assumption, we impose $\mu_{1}=-\tfrac{1}{2} \sig_{1}^2$ so as to ensure that
$S=\ee^X$ is a martingale. 
For many models in finance, the dimension of the diffusion is $d=1$ (e.g., CEV) or $d=2$ (e.g., Heston, SABR).
For the special cases $d=1,2$, we write {$\Ac(t)$ as} 
\begin{align}
\Ac(t)
    &=  a(t,x,y) ( \d_x^2 - \d_x ) + f(t,x,y) \d_{y} + b(t,x,y) \d_y^2 + c(t,x,y) \d_x \d_y ,\qquad (x,y)\in\mathbb{R}^{2} \label{eq:A-2d}
\end{align}
where
\begin{align}
 a &:=  \frac{\sig_{1}^2}{2} ,  &
f   &:= \mu_2 , &
b &:=  \frac{\sig_{2}^2}{2} ,&
 c  & :=  \rho \sig_{1} \sig_{2}. \label{eq:abc}
\end{align}
When $d=1$ (i.e., local volatility models) only $a$ appears.
\begin{remark}[Deterministic interest rates] For deterministic
interest rates $r(t)$ one must compute expectations of the form
\begin{align}
\ut(t,\xt,y)
    &:= \Eb \[\ee^{- \int_t^T r(s) \dd s}  \varphi(\Xt_T) | \Xt_t = \xt , Y_t = y \] , &
    &\text{where}&
\dd \Xt_t
    &= \dd X_t + r(t) \dd t.
\end{align}
In this case a simple change of variables
\begin{align}
u(t,x(t,\xt),y)
    &:= \ee^{\int_t^T r(s)} \ut(t,\xt,y)  , &
x(t,\xt)
    &:= \xt + \int_t^T r(s) \dd s, \label{eq:new.u}
\end{align}
reveals that the function $u$, as defined as in \eqref{eq:new.u}, satisfies \eqref{eq:v.pde}.
\end{remark}

\subsection{Polynomial expansions of $\Ac(t)$}
\label{sec:approximating}
We note that \eqref{eq:A} is a special case of the more general $d$-dimensional second order differential operator
\begin{align}\label{operator_AA}
 \Ac(t) &= \sum_{i,j=1}^{d}a_{ij}(t,z)\p_{z_{i}z_{j}}+\sum_{i=1}^{d}a_{i}(t,z)\p_{z_{i}},\qquad t\in\R_+,\ z\in
\mathbb{R}^d .
\end{align}
Equivalently, we can also write the operator $\Ac(t)$ in a more compact form, i.e.
\begin{align}
\Ac(t)
    &:= \sum_{|\alpha |\leq 2} a_{\alpha}(t,z) D^{\alpha}_{z} , \qquad
t
    \in \Rb_+,
z
    \in \mathbb{R}^d , \label{operator_A}
\end{align}
where, using standard multi-index notation we have
\begin{align}
\alpha
    &=  (\alpha_1,\cdots,\alpha_d)\in \mathbb{N}^{d}_{0}, &
|\alpha|
    &=  \sum_{i=1}^{d}\alpha_i, &
D_{z}^{\alpha}
    =       \partial^{\alpha_1}_{z_1}\cdots \partial^{\alpha_d}_{z_d} ,
\end{align}
In this section we introduce a family of expansion schemes for {$\Ac(t)$}, which we shall use to construct closed-form approximate solutions (one for each family) of Cauchy problem \eqref{eq:v.pde}.
\begin{definition}
\label{def:expansion}
Let $N\in\mathbb{N}_0$. We say that 
$(\Ac_n(t))_{0\leq n\leq N}$ is {\it an $N$th order polynomial expansion 
}
if
\begin{align}
\Ac_n(t,z)\equiv\Ac_n(t)
    :=  \sum_{|\alpha |\leq 2}  a_{\alpha,n}(t,z) D_z^{\alpha} \label{eq:A.expand_bis}
\end{align}
where
\begin{enumerate}
\item[(i)] for any $t \in [0,T]$ the functions $a_{\alpha,n}(t,\cdot)$ are polynomials, and for any $z\in\Rb^{d}$ the functions $a_{\alpha,n}(\cdot,z)$ belong to $L^{\infty}([0,T])$,
\item[(ii)] for any $t \in [0,T]$ we have $a_{\alpha,0}(t,\cdot) =a_{\alpha,0}(t)$, and the constant-in-space coefficients second order operator $\Ac_0(t)$
is elliptic.
\end{enumerate}
\end{definition}
\noindent
The idea behind our approximation method is to choose a polynomial expansion such that the
sequences of partial sums $\sum_{n=0}^N a_{\alpha,n}(t)$ approximate the coefficients
$a_{\alpha}(t,z)$, either pointwise or in some norm.
Below, we present some examples.
\begin{example}[Taylor polynomial expansion]
\label{example:Taylor}
Assume the coefficients $a_{\alpha}(t,\cdot)\in C^N(\mathbb{R}^d)$. Then, for any fixed
$\bar{z}\in\Rb^{d}$, $n\leq N$, we define $a_{\alpha,n}$ as the $n$-th order term of the Taylor
expansion of $a_{\alpha}$ in the spatial variables around $\zb$.  That is, we set
\begin{align}\label{eq:taylor_polyn}
a_{\alpha,n}(\cdot,z)
  &=        \sum_{|\beta|=n}\frac{D_z^{\beta}a_{\alpha}(\cdot,\bar{z})}{\beta!}(z-\bar{z})^{\beta}, &
n
    &\leq N, &
|\alpha|
    &\leq 2 ,
\end{align}
where as usual {$\beta!=\beta_{1}!\cdots\beta_{d}!$ and $z^\beta = z_1^{\beta_1} \cdots z_d^{\beta_d}$}.
\end{example}

\begin{example}[Time-dependent Taylor polynomial expansion]
\label{example:TimeTaylor}
Assume the coefficients $a_{\alpha}(t,\cdot)\in C^N(\mathbb{R}^d)$. Then, for any fixed
$\bar{z}:\Rb_+ \to \R^{d}$, we define $a_{\alpha,n}$ as the $n$-th order term of the Taylor
expansion of $a_{\alpha}$ in the spatial variables around $\zbar(\cdot)$.  That is, we set
\begin{align}
 a_{\alpha,n}(\cdot,z)
    &=\sum_{|\b|=n}\frac{D_z^{\b}a_{\alpha}(\cdot,\bar{z}(\cdot))}{\b!}(z-\bar{z}(\cdot))^{\b}, &
n
    &\leq N, &
|\alpha|
    &\leq 2 .
\end{align}
\end{example}

\begin{example}[Hermite polynomial expansion]\label{example:Hilbert}
Hermite expansions can be useful when the diffusion coefficients are not smooth.  A remarkable
example in financial mathematics is given by the Dupire's local volatility formula for models with
jumps (see \cite{frizyor2013}). In some cases, e.g., the well-known Variance-Gamma model, the
fundamental solution (i.e., the transition density of the underlying stochastic model) has
singularities.  In such cases, it is natural to approximate it in some $L^{p}$ norm rather than in
the pointwise sense. For the Hermite expansion centered at $\zb$, one sets
\begin{align}
a_{\alpha,n}(t,z)
    &=  \sum_{|\b|=n} \< \Hv_\b(\cdot-\bar{z}) , a_{\alpha}(t,\cdot) \>_{\Gam} \Hv_\b(z-\bar{z}),&
n
    &\geq 0, &
|\alpha|
    &\leq 2 .
\end{align}
where the inner product $\<\cdot,\cdot\>_\Gam$ is an integral over $\Rb^d$ with a Gaussian
weighting centered at $\zb$ and the functions $\Hv_\beta(z) = H_{\beta_1}(z_1) \cdots
H_{\beta_d}(z_d)$ where $H_n$ is the $n$-th one-dimensional Hermite polynomial (properly
normalized so that $\< \Hv_\alpha, \Hv_\beta \>_\Gam = \del_{\alpha,\beta}$ with
$\del_{\alpha,\beta}$ being the Kronecker's delta function).
\end{example}

%
%

\subsection{Formal solution}
\label{sec:dyson}
In this section, we introduce a heuristic procedure to construct an approximate solution of the backward Cauchy problem \eqref{eq:v.pde}.  Hereafter we will explicitly indicate $t$-dependence in all operators.  On the other hand, we will generally hide $z$-dependence, except where it is needed for clarity.

Let us consider a polynomial expansion $(\Ac_n(t))_{n\geq 0}$, and assume that the operator $\Ac(t)$ in \eqref{operator_A} can be formally written as
\begin{align}
\Ac(t)
    &=  \Ac_0(t) + \Bc(t) , &
\Bc(t)
    &=  \sum_{n=1}^\infty \Ac_n(t), \label{eq:A.expand} 
\end{align}
Inserting expansion \eqref{eq:A.expand} for $\Ac(t)$ into Cauchy problem \eqref{eq:v.pde} we find
\begin{align}
( \d_t + \Ac_0(t) ) u(t)
    &=  - \Bc(t) u(t) , &
u(T)
    &=\varphi .
\end{align}
By Duhamel's principle, we have
\begin{align}
u(t)
    &=  \Pc_0(t,T)\varphi + \int_t^T \dd t_1 \, \Pc_0(t,t_1) \Bc(t_1) u(t_1) , \label{eq:u.duhemel}
\end{align}
where $\Pc_0(t,T)$ is the \emph{semigroup} of operators generated by $\Ac_0(t)$; we will explicitly
define $\Pc_0(t,T)$ in \eqref{eq:u0}.
Inserting expression \eqref{eq:u.duhemel} for $u$ into the right-hand side of \eqref{eq:u.duhemel} and iterating we obtain
\begin{align}
u(t_0)
    &=  \Pc_0(t_0,T)\varphi
            + \int_{t_0}^T \dd t_1 \, \Pc_0(t_0,t_1) \Bc(t_1) \Pc_0(t_1,T)\varphi \\ & \qquad
            + \int_{t_0}^T \dd t_1 \int_{t_1}^T \dd t_2 \, \Pc_0(t_0,t_1) \Bc(t_1)
                \Pc_0(t_1,t_2) \Bc(t_2) u(t_2) \\
    &=  \cdots \\
    &=  \Pc_0(t_0,T)\varphi + \sum_{k=1}^\infty
            \int_{t_0}^T \dd t_1 \int_{t_1}^T \dd t_2 \cdots \int_{t_{k-1}}^T \dd t_k
            \\ & \qquad
            \Pc_0(t_0,t_1) \Bc(t_1)
            \Pc_0(t_1,t_2) \Bc(t_2) \cdots
            \Pc_0(t_{k-1},t_k) \Bc(t_k)
            \Pc_0(t_k,T)\varphi
            \label{eq:dyson} \\
    &=  \Pc_0(t_0,T)\varphi + \sum_{n=1}^\infty \sum_{k=1}^n
            \int_{t_0}^T \dd t_1 \int_{t_1}^T \dd t_2 \cdots \int_{t_{k-1}}^T \dd t_k
            \\ & \qquad \sum_{i \in I_{n,k}}
            \Pc_0(t_0,t_1) \Ac_{i_1}(t_1)
            \Pc_0(t_1,t_2) \Ac_{i_2}(t_2) \cdots
            \Pc_0(t_{k-1},t_k) \Ac_{i_k}(t_k)
            \Pc_0(t_k,T)\varphi,
            \label{eq:u.dyson} \\
I_{n,k}
    &= \{ i = (i_1, i_2, \cdots , i_k ) \in \mathbb{N}^k : i_1 + i_2 + \cdots + i_k = n \}.
            \label{eq:Ink}
\end{align}
To obtain \eqref{eq:u.dyson} from \eqref{eq:dyson} we have used the fact that from \eqref{eq:A.expand} the operator $\Bc(t)$ is an infinite sum, and we have partitioned on the sum $(i_1 + i_2 + \cdots + i_k)$ of the subscripts of the $(\Ac_{i_k}(t))$.  In light of expansion \eqref{eq:u.dyson} we set
\begin{align}
u
    &=  \sum_{n=0}^\infty u_n , \label{eq:u.sum}
\end{align}
where we have defined
\begin{align}
u_0(t_0)
    &:= \Pc_0(t_0,T) \varphi, \label{eq:u0.def} \\
u_n(t_0)
    &:= \sum_{k=1}^n
            \int_{t_0}^T \dd t_1 \int_{t_1}^T \dd t_2 \cdots \int_{t_{k-1}}^T \dd t_k
            \\ & \qquad \sum_{i \in I_{n,k}}
            \Pc_0(t_0,t_1) \Ac_{i_1}(t_1)
            \Pc_0(t_1,t_2) \Ac_{i_2}(t_2) \cdots
            \Pc_0(t_{k-1},t_k) \Ac_{i_k}(t_k)
            \Pc_0(t_k,T) \varphi\label{eq:un.def} .
\end{align}
\subsection{Expression for $u_0$}\label{sec:u0}
By assumption, the functions {$a_{\alpha,0}$} depend only on $t$.  Therefore, the operator $\Ac_{0}(t)$ is
the generator of a diffusion with time-dependent parameters.  It will be useful to write the operator $\Ac_0(t)$ in the following form:
\begin{align}\label{def_A0_bis}
\Ac_0(t)
    &=\frac{1}{2} \sum_{i,j=1}^{d} C_{ij}(t)\p_{z_{i}z_{j}} + \< m(t), \nabla_{z} \>, &
\< m(t), \nabla_{z} \>
    &=  { \sum_{i=1}^d m_i(t) \d_{z_i} } .
\end{align}
Here the $d\times d$-matrix $C(t)$ is positive definite, for any $t\in[0,T]$, and $m$ is a $d$-dimensional vector.
The action of the semigroup of operators $\Pc_0(t_0,T)$ generated by $\Ac_0(t)$ is well-known.  For any measurable function $\varphi$ that is at most exponentially growing we have
\begin{align}
u_0(t_0)
    :=  \Pc_0(t_0,T) \varphi    &=  \int_{\Rb^d}  \Gam_{0}(t,\cdot;T,\zeta)\varphi(\zeta) \,\dd \zeta , \label{eq:u0}
\end{align}
where $\Gam_0(t,z;T,\zeta)$ is the $d$-dimensional Gaussian density
\begin{align}
\Gam_{0}(t,z;T,\zeta)
        &=  \frac{1}{  \sqrt{(2\pi)^{d}|\Cv(t,T)|} }
    \exp\left(-\frac{1}{2}\langle\Cv^{-1}(t,T) (\zeta - z-\mv(t,T)),
    (\zeta -z-\mv(t,T))\rangle\right) \label{e22and}
\end{align}
with covariance matrix $\Cv(t,T)$ and mean vector $z+\mv(t,T)$ given by:
\begin{align}
 \Cv(t,T)
    &=  \int_t^T \dd s \, C(s) , &
\mv(t,T)
    &=  \int_t^T \dd s \, m(s) . \label{eq:covariance_mean}
\end{align}
Note that the function $u_0$ as it is defined in \eqref{eq:u0} is the unique non-rapidly
increasing solution of the homogeneous backward Cauchy problem $(\d_t + \Ac_0(t))u_0=0$ with
terminal condition $u_0(T)=\varphi$.

\subsection{Expression for $u_n$}
Remarkably, as the following theorem shows, every $u_n(t)$ can be expressed as a differential
operator $\Lc_n(t,T)$ acting on $u_0(t)$.
\begin{theorem}
\label{thm:dyson}
Assume $\varphi \in \Sc(\Rb^d)$, the Schwartz space of rapidly decreasing functions on $\Rb^d$.
Then the function $u_n$ defined in \eqref{eq:un.def} is given explicitly by
\begin{align}
u_n(t_0)
    &=  \Lc_n(t_0,T) u_0(t_0) , \label{eq:un}
\end{align}
where $u_0$ is given by \eqref{eq:u0} and
\begin{align}
\Lc_n(t_0,T)
    &=  \sum_{k=1}^n
            \int_{t_0}^T \dd t_1 \int_{t_1}^T \dd t_2 \cdots \int_{t_{k-1}}^T \dd t_k
            \sum_{i \in I_{n,k}}
            \Gc_{i_1}(t_0,t_1)
            \Gc_{i_2}(t_0,t_2) \cdots
            \Gc_{i_k}(t_0,t_k) , \label{eq:Ln}
\end{align}
with $I_{n,k}$ as defined in \eqref{eq:Ink}, 
\begin{align}
\Gc_i(t_0,t_k) &:= \Ac_{i}(t_k,\Mc(t_0,t_k))
= \sum_{|\alpha |\leq 2}  a_{\alpha,i}(t_k,\Mc(t_0,t_k)) D_z^\alpha ,   \label{eq:G.def}
\end{align}
with $\Ac_{i}(t,z)$ as in \eqref{eq:A.expand_bis}, and
\begin{align}
\Mc(t,s)
    &:= z+\mv(t,s)+ \Cv(t,s)\nabla_{z}. \label{eq:M.def}
\end{align}
\end{theorem}
\begin{proof}
The main idea of the proof is to show that the operator $\Gc_i(t_0,t_k)$ in \eqref{eq:G.def} satisfies
\begin{align}
\Pc_0(t_0,t_k) \Ac_{i}(t_k)
    &=  \Gc_i(t_0,t_k) \Pc_0(t_0,t_k) . \label{eq:PA=GP}
\end{align}
Assuming \eqref{eq:PA=GP} holds, we can use the fact that $\Pc_0(t_k,t_{k+1})$ is a semigroup
\begin{align}
\Pc_0(t_0,T)
    &=  \Pc_0(t_0,t_1) \Pc_0(t_1,t_2) \cdots \Pc_0(t_{k-1},t_k) \Pc_0(t_k,T) ,
\end{align}
and we can re-write \eqref{eq:un.def} as
\begin{align}
u_n(t_0)
    &=  \sum_{k=1}^n
            \int_{t_0}^T \dd t_1 \int_{t_1}^T \dd t_2 \cdots \int_{t_{k-1}}^T \dd t_k
            \sum_{i \in I_{n,k}}
            \Gc_{i_1}(t_0,t_1)
            \Gc_{i_2}(t_0,t_2) \cdots
            \Gc_{i_k}(t_0,t_k)
            \Pc_0(t_0,T) \varphi , \label{eq:u.dyson2}
\end{align}
from which \eqref{eq:un}-\eqref{eq:Ln} follows directly.  Thus, we only need to show that
$\Gc_i(t_0,t_k)$ satisfies \eqref{eq:PA=GP}. 
The condition $\varphi \in \Sc(\Rb^d)$ guarantees that $u_0(t,\cdot)$ belongs to the Schwartz class of rapidly decaying
functions for all $t<T$.  Therefore, any function of the form $p(z) D_z^\beta u_0(t,z)$, where $p$ is a polynomial, has a
Fourier representation.  Thus, without loss of generality, we can investigate how the operator
$\Pc_0(t_0,t_k) \Ac_{i}(t_k)$ acts on the oscillating exponential
$\ee_\lam(x):= \ee^{ \ii \< \lam , x\> }$.
We note that
\begin{align}
\Pc_0(t_0,t_k) \ee_\lam(z)
    &=  \ee^{\Phi_0(t_0,t_k,\lam)}\ee_\lam(z) ,         &
    &\text{where}&
\Phi_0(t_0,t_k,\lam)
    &=  \sum_{|\alpha |\leq 2} (\ii \lam)^\alpha \int_{t_0}^{t_k} \dd t \, a_{\alpha,0}(t).\label{eq:P.exp}
\end{align}
Next, we observe that the operator $\Mc_i(t_0,t_k)$, the $i$-th component of $\Mc(t_0,t_k)$ in \eqref{eq:M.def} can be written
\begin{align}
\Mc_i(t_0,t_k)
    &=  M_i(t_0,t_k, -\ii \nabla_z ) , &
M_i(t_0,t_k, \lam )
    &=  - \ii \d_{\lam_i} \( \Phi_0(t_0,t_k,\lam) + \ii \< \lam, z \> \) . \label{eq:M}
\end{align}
Using \eqref{eq:M} we observe that for any natural number $n$ we have
\begin{align}
(-\ii \d_{\lam_i})^n \ee^{\Phi_0(t_0,t_k,\lam)}\ee_\lam(z)
    &=  (-\ii \d_{\lam_i})^{n-1} M_i(t_0,t_k,\lam )
            \ee^{\Phi_0(t_0,t_k,\lam)}\ee_\lam(z) \\
    &=  \Mc_i(t_0,t_k) (-\ii \d_{\lam_i})^{n-1}
            \ee^{\Phi_0(t_0,t_k,\lam)}\ee_\lam(z) \\
    &=  \cdots \\
    &=  [ \Mc_i(t_0,t_k) ]^n \ee^{\Phi_0(t_0,t_k,\lam)}\ee_\lam(z) .
\end{align}
Noting the $\d_{\lam_i}$ and $\d_{\lam_j}$ commute, it is clear that $\Mc_i(t_0,t_k)$ and $\Mc_j(t_0,t_k)$ also commute.
Thus, for any multi-index $\beta$ we have
\begin{align}
(-\ii \nabla_\lam)^\beta \ee^{\Phi_0(t_0,t_k,\lam)}\ee_\lam(z)
    &=  (\Mc(t_0,t_k))^\beta \ee^{\Phi_0(t_0,t_k,\lam)}\ee_\lam(z) , \label{eq:a.M}
\end{align}
Finally, we compute
\begin{align}
\Pc_0(t_0,t_k) \Ac_{i}(t_k) \ee_\lam(z)
    &=  \sum_{|\alpha |\leq 2}  \Pc_0(t_0,t_k)
            a_{\alpha,i}(t_k,z) D^{\alpha}_{z} \ee_\lam(z) &
    &\text{(by \eqref{eq:A.expand})} \\
    &=  \sum_{|\alpha |\leq 2}  (\ii \lam)^{\alpha} \Pc_0(t_0,t_k)
            a_{\alpha,i}(t_k,z) \ee_\lam(z) \\
    &=  \sum_{|\alpha |\leq 2}  (\ii \lam)^{\alpha} a_{\alpha,i}(t_k,-\ii \nabla_\lam)
            \Pc_0(t_0,t_k) \ee_\lam(z) \\
    &=  \sum_{|\alpha |\leq 2}  (\ii \lam)^{\alpha} a_{\alpha,i}(t_k,-\ii \nabla_\lam)
            \ee^{\Phi_0(t_0,t_k,\lam)}\ee_\lam(z) &
    &\text{(by \eqref{eq:P.exp})} \\
    &=  \sum_{|\alpha |\leq 2}  (\ii \lam)^{\alpha} a_{\alpha,i}(t_k,\Mc(t_0,t_k))
            \ee^{\Phi_0(t_0,t_k,\lam)}\ee_\lam(z) &
    &\text{(by \eqref{eq:a.M})} \\
    &=  \sum_{|\alpha |\leq 2}  a_{\alpha,i}(t_k,\Mc(t_0,t_k)) D_z^\alpha
            \ee^{\Phi_0(t_0,t_k,\lam)}\ee_\lam(z)   \\
    &=  \sum_{|\alpha |\leq 2}  a_{\alpha,i}(t_k,\Mc(t_0,t_k)) D_z^\alpha
            \Pc_0(t_0,t_k) \ee_\lam(z)   &
    &\text{(by \eqref{eq:P.exp})} \\
    &=  \Gc_i(t_0,t_k) \Pc_0(t_0,t_k) \ee_\lam(z) , &
    &\text{(by \eqref{eq:A.expand_bis} and \eqref{eq:G.def})}
\end{align}
which concludes the proof.
\end{proof}
\begin{remark}[Call payoffs]
\label{rmk:call} {As we will show in Section \ref{sec:price.error} the functions $(u_n)$ can be
alternatively characterized as solutions of a nested sequence of Cauchy problems (see equation
\eqref{eq:v.n.pide.xbar} for the case when $(\Ac_n(t))$ is expanded in a Taylor series as in
Example \ref{example:Taylor}).  One can check directly that when $\varphi(x) = (\ee^x - \ee^k)$,
the functions $(u_n)$ with each $u_n$ given by $u_n(t) = \Lc_n(t,T) u_0(t)$ satisfy the nested
Cauchy problems.  Thus, Theorem \ref{thm:dyson} also holds for Call option payoffs.  This is true
for any expansion $(\Ac_n(t))$ satisfying Definition \ref{def:expansion}.}
\end{remark}
\begin{remark}
\label{rmk:terms}
The number of terms in $\Lc_n(t,T)$ grows faster than $n!$, which presents a computational challenge for large $n$.  Nevertheless, we shall see in the numerical example provided in Section \ref{sec:examples} that excellent approximations can be achieved with $n = 3$.
\end{remark}
\begin{remark}
When $d=1,2$, the operator $\Ac(t)$ is given by \eqref{eq:A-2d}.  In this case, we write $\Ac_i(t)$ as
\begin{align}
\Ac_i(t)
    &:= a_{i}(t,x,y) (\d_x^2-\d_x)
            + f_i(t,x,y) \d_y
            + b_i(t,x,y) \d_y^2
            + c_i(t,x,y) \d_x \d_y,
\end{align}
and we have explicitly
\begin{align}
\Gc_i(t,s)
    &:= a_{i}\(s,\Mc_x(t,s),\Mc_y(t,s)\) (\d_x^2-\d_x)
            + f_i\(s,\Mc_x(t,s),\Mc_y(t,s)\) \d_y \\ &\qquad
            + b_i\(s,\Mc_x(t,s),\Mc_y(t,s)\) \d_y^2
            + c_i\(s,\Mc_x(t,s),\Mc_y(t,s)\) \d_x \d_y, \label{eq:G.lsv} \\
\Mc_x(t,s)
    &=  x - \int_t^s \dd q \,  a_0(q) + 2 \int_t^s \dd q \,  a_0(q) \d_x + \int_t^s \dd q \,  c_0(q) \d_y , \\
\Mc_y(t,s)
    &=  y + \int_t^s \dd q \,  f_0(q) + 2 \int_t^s \dd q \,  b_0(q) \d_y + \int_t^s \dd q \,  c_0(q) \d_x .
\end{align}
\end{remark}

%
%


%
%

\section{Implied volatility expansion}
\label{sec:impvol}
In this section, we derive an explicit implied volatility approximation from the asymptotic
pricing expansion developed in the previous section.  To begin our analysis, we fix a multifactor
LSV model for $X = \log S$ as in \eqref{eq:StochVol}, a time $t$, a maturity date $T>t$, the
initial values {$(X_t,Y_t)=(x,y)\in\mathbb{R}\times\mathbb{R}^{d-1}$} and a Call option
payoff $\varphi(X_T)=(\ee^{X_T} - \ee^k)^+$. Our goal is to find the implied volatility for
\emph{this particular Call option}. To ease notation, we will sometimes suppress the dependence on
$(t,T,x,y,k)$. However, the reader should keep in mind that the implied volatility of the option
under consideration \emph{does} depend on $(t,T,x,y,k)$, even if this is not explicitly indicated.
Below, we provide definitions of the \emph{Black-Scholes price} and \emph{implied volatility},
which will be fundamental throughout this section.
\begin{definition}\label{def:BS}
The \emph{Black-Scholes price} $u^{\BS}$
is given by
\begin{align}
u^{\BS}(\s;\t,x,k)
    &:= \ee^x \Nc(d_{+}) - \ee^k \Nc(d_{-}) , &
d_{\pm}
    &:= \frac{1}{\sig \sqrt{\tau}} \( x - k \pm \frac{\sig^2 \tau}{2}  \) , &
\tau
        &:= T-t , \label{eq:BS}
\end{align}
where $\Nc$ is the CDF of a standard normal random variable.
\end{definition}
\begin{remark}
\label{rmk:u0=uBS}  It follows from \eqref{eq:u0} that when $\varphi(x)=(\ee^x-\ee^k)^+$ we have
\begin{align}
u_0(t,x)
    &= u^\BS(\sig_0;T-t,x,k) , &
    &\text{where}&
\sig_0
    &=  \sqrt{ \frac{2}{T-t} \int_t^T a_{0}(s) \dd s},\label{eq:sigma_0}
\end{align}
where $a_0=C_{1,1}$ as in \eqref{def_A0_bis}, or according to the multi-index notation, $a_0=a_{(2,0,\ldots,0),0}$.
\end{remark}
\begin{definition}
\label{def:imp.vol.def} {For fixed $(\t,x,k)$}, the \emph{implied volatility} corresponding to
a Call price $u\in\,((\ee^{x}-\ee^{k})^+,\ee^x)$ is defined as the unique strictly positive real
solution $\sig$ of the equation
\begin{align}
 u^{\BS}(\sig;\t,x,k)    &= u.   \label{eq:imp.vol.def}
\end{align}
\end{definition}

\subsection{Formal derivation}\label{subsec:imp_vol_formal}
We present here a formal derivation of our implied volatility expansion, which is based on the price expansion presented in Section \ref{sec:model}.
Throughout this section {$(t,T,x,k)$ {\it are fixed} and thus we use the short notation
$$u^{\BS}(\s)=u^{\BS}(\s;T-t,x,k)$$} for the Black-Scholes price. Consider the family of approximate Call
prices indexed by $\del$
\begin{align}
 u(\delta)
    &=  \sum_{n=0}^N \delta^n  u_{n}
     =  u^{\BS}(\sigma_0)+\sum_{n=1}^N \delta^n u_{n} , &
\del
    &\in    [0,1] , \label{eq:u.expand.again}
\end{align}
with $\sigma_0$ as in \eqref{eq:sigma_0} and the functions {$u_n(t)=\Lc_n(t,T) u_0(t)$} as given
in Theorem \ref{thm:dyson}. Note that setting $\del=1$ yields our price expansion.
Defining
\begin{align}
g(\del)
    &:=  (u^\BS)^{-1}(u(\del)),&\del
    &\in    [0,1] . \label{eq:matt}
\end{align}
we seek the implied volatility $\sig=g(1)$.
We will show in Section \ref{sec:error.impvol}, Lemma \ref{lem:imp_vol_u_delta}, that under
suitable assumptions $u(\delta)\in\,((\ee^{x}-\ee^{k})^+,\ee^x)$ for any $\del\in[0,1]$. This
guarantees that $g(\del)$ in \eqref{eq:matt} exists. 
By expanding both sides of \eqref{eq:matt} as a Taylor series in $\del$, we see that $\sigma$
admits an expansion of the form
\begin{align}
\sig=g(1)
  &=  \sig_0 + \sum_{n=1}^\infty  \sig_{n}, &
\sig_{n}
    &=  \frac{1}{n!} \d_{\delta}^{n} g(\delta) \vert_{\delta=0}
   . \label{ae6}
\end{align}
Note that, by \eqref{eq:u.expand.again} we also have
\begin{align}\label{ae8}
u_{n}
    &=  \frac{1}{n!}\partial_{\delta}^{n}u^{\BS}(g(\delta))\vert_{\delta=0} , &
1
    &\leq n \leq N .
\end{align}
The right-hand side of \eqref{ae8} can be computed by
applying the Bell polynomial version of the Faa di Bruno's formula, which is given in Appendix \ref{append:faa_bell}: 
\begin{align}\label{ae7}
 u_{n}
    &=  \frac{1}{n!}\sum_{h=1}^{n}
            \partial^{h}_{\sig}u^{\BS}(\sig_{0})
            \mathbf{B}_{n,h}\left(\partial_{\delta}g(\delta),\partial_{\delta}^{2}g(\delta),\dots,
            \partial_{\delta}^{n-h+1}g(\delta)\right)\vert_{\delta=0} , &
1
    &\leq n \leq N .
\end{align}
Combining \eqref{ae7} with \eqref{ae6}, one can solve for $\sig_n$ explicitly in terms of
$(\sig_k)_{0 \leq k \leq n-1}$, which yields
\begin{align}
\sig_n
    =  \frac{u_n}{\d_\sig u^\BS(\sig_0)}
            - \frac{1}{n!} \sum_{h=2}^{n}
       \mathbf{B}_{n,h}\(1!\,\sig_{1},2!\,\sig_{2},\dots,(n-h+1)!\, \sig_{n-h+1}\)
            \frac{\d_\sig^h u^\BS(\sig_0)}{\d_\sig u^\BS(\sig_0)}, \qquad
1
    \leq n \leq N . \label{eq:sig.n_bis}
\end{align}
Note that expression \eqref{eq:sig.n_bis} for $\sig_n$ involves two sorts of terms: $u_n/\d_\sig
u^\BS(\sig_0)$ and $\d_\sig^n u^\BS(\sig_0)/\d_\sig u^\BS(\sig_0)$.  We will prove that these
terms can be computed explicitly without any numerical integration or special functions.  The
proof will rely on the following lemma.
\begin{lemma}
\label{lemma:ratio}
Let $m\geq 0$ and fix $(t,T,k,\sig_0)$.  Then
\begin{align}
\frac{\d_x^m (\d_x^2 - \d_x ) u^{\BS}(\sig_0)}{(\d_x^2 - \d_x ) u^{\BS}(\sig_0)}
    = \(-\frac{1}{\sig_0\sqrt{2\tau}}\)^{m} H_{m}\(\zeta\) , \qquad
\zeta
    := \frac{x-k-\frac{1}{2}\sig_0^2 \tau}{\sig_0\sqrt{2\tau}} , \qquad
\tau
    := T-t, \label{eq:result}
\end{align}
where $H_n(\zeta) := (-1)^n \ee^{\zeta^2} \d_\zeta^n \ee^{-\zeta^2}$ is the $n$-th Hermite polynomial.
\end{lemma}
\begin{proof}
Using the Black-Scholes formula \eqref{eq:BS}, a direct computation shows
\begin{align}
(\d_x^2 - \d_x ) u^{\BS}(\sig_0)
    &=  \frac{1}{\sig_0\sqrt{2\pi\tau}}\ee^{-\zeta^2+k} ,
\end{align}
with $\zeta = \zeta(x)$ as above.  Hence
\begin{align}
\frac{\d_x^m (\d_x^2 - \d_x ) u^{\BS}(\sig_0)}{(\d_x^2 - \d_x ) u^{\BS}(\sig_0)}
    &=  \ee^{\zeta^2}\d_x^m \ee^{-\zeta^2}
    =       \(\frac{1}{\sig_0\sqrt{2\tau}}\)^m\ee^{z^2}\d_\zeta^m \ee^{-\zeta^2}
    =       \(\frac{-1}{\sig_0\sqrt{2\tau}}\)^m H_m(\zeta) ,
\end{align}
where in the last equality we have used the definition of the $m$th Hermite polynomial, recalled above.
\end{proof}
\begin{proposition}
\label{prop.sig.k}
Fix $(t,T,k,\sig_0)$ and let $\zeta$ and $\tau$ be as in Lemma \ref{lemma:ratio}.  Then for any $n \geq 2$ we have
\begin{align}
\frac{\d_\sig^n u^{\BS}(\sig_0)}{\d_\sig u^{\BS}(\sig_0)}
    &=  \sum_{q=0}^{\left\lfloor n/2 \right\rfloor}\sum_{p=0}^{n-q-1}
            c_{n,n-2q}\sig_0^{n-2q-1} \tau^{n-q-1} \binom{n-q-1}{p}
            \(-\frac{1}{\sig_0\sqrt{2\tau}}\)^{p+n-q-1} H_{p+n-q-1}(\zeta), \label{eq:2}
\end{align}
where the coefficients $(c_{n,n-2k})$ are defined recursively by
\begin{align}
c_{n,n}
    &=1 , &
    &\text{and}&
c_{n,n-2q}
    &=(n-2q+1) c_{n-1,n-2q+1} + c_{n-1,n-2q-1}, &
q &\in \{ 1, 2, \cdots , \left\lfloor n/2 \right\rfloor \}.
\end{align}
\end{proposition}
\begin{proof}
Define the operator $\Jc := \tau (\d_x^2 - \d_x)$.  It is classical that $\d_\sig u^\BS(\sig_0) = \sig_0 \Jc u^\BS(\sig_0)$.
We claim that the following identity holds for any $n\in\mathbb{N}$
\begin{align}
\d_\sig^n u^\BS(\sig_0)
    &= \sum_{q=0}^{\left\lfloor n/2 \right\rfloor} c_{n,n-2q}\sig_0^{n-2q}\Jc^{n-q} u^\BS(\sig_0) , \label{eq:claim}
\end{align}
where $c_{n,n}=1$ and $c_{n,n-2q} =(n-2q+1) c_{n-1,n-2q+1} + c_{n-1,n-2q-1}$ for any integer $q \in \{ 1, 2, \cdots , \left\lfloor n/2 \right\rfloor \}$. The proof of \eqref{eq:claim} is a simple yet tedious recursion relation, which we omit for brevity.  Now, we compute
\begin{align}
\frac{\d_\sig^n u^\BS(\sig_0)}{\d_\sig u^\BS(\sig_0)}
    &=  \sum_{q=0}^{\left\lfloor n/2 \right\rfloor} c_{n,n-2q}\sig_0^{n-2q}
            \frac{\Jc^{n-q} u^\BS(\sig_0)}{\d_\sig u^\BS(\sig_0)}
    =   \sum_{q=0}^{\left\lfloor n/2 \right\rfloor} c_{n,n-2q}\sig_0^{n-2q} \tau^{n-q}
            \frac{(\d_x^2-\d_x)^{n-q} u^\BS(\sig_0)}{\d_\sig u^\BS(\sig_0)} \\
    &=  \sum_{q=0}^{\left\lfloor n/2 \right\rfloor} c_{n,n-2q}\sig_0^{n-2q} \tau^{n-q}
            \frac{(\d_x^2-\d_x)^{n-q-1}(\d_x^2-\d_x) u^\BS(\sig_0)}{\tau \sig_0 (\d_x^2-\d_x) u^\BS(\sig_0)} \\
    &=  \sum_{q=0}^{\left\lfloor n/2 \right\rfloor}\sum_{p=0}^{n-q-1}
            c_{n,n-2q}\sig_0^{n-2q-1} \tau^{n-q-1} \binom{n-q-1}{p}
            \frac{\d_x^{p+n-q-1}(\d_x^2-\d_x) u^\BS(\sig_0)}{(\d_x^2-\d_x) u^\BS(\sig_0)}\\
    &=  \sum_{q=0}^{\left\lfloor n/2 \right\rfloor}\sum_{p=0}^{n-q-1}
            c_{n,n-2q}\sig_0^{n-2q-1} \tau^{n-q-1} \binom{n-q-1}{p}
            \(-\frac{1}{\sig_0\sqrt{2\tau}}\)^{p+n-q-1} H_{p+n-q-1}(\zeta),
\end{align}
where to obtain the last equality we have used \eqref{eq:result}.
\end{proof}
\begin{proposition}
\label{prop.un}
Fix $(t,T,x,y)$.  For every polynomial expansion $(\Ac_n(t))$ satisfying Definition \ref{def:expansion}
and for every $n \in \mathbb{N}$, the ratio ${u_n}/{\d_\sig u^\BS(\sig_0)}$ is a finite sum of the form
\begin{align}
\frac{u_n}{\d_\sig u^\BS(\sig_0)}
    &=  \sum_m \chi_m^{(n)}(t,T,x,y) \(-\frac{1}{\sig_0\sqrt{2\tau}}\)^{m} H_{m}\(\zeta\) , \label{eq:form}
\end{align}
where $\zeta$ and $\tau$ are as in Lemma \ref{lemma:ratio}. The coefficients $\chi_m^{(n)}(t,T,x,y)$
are explicit function of $x$ and $y$ and contain iterated integrals in the time-variable.  If the
iterated time-integrals can be computed explicitly then $\chi_m^{(n)}(t,T,x,y)$ is explicit in all
variables.
\end{proposition}
\begin{proof}
From equation \eqref{eq:G.lsv} and Remark \ref{rmk:u0=uBS} we observe that, for the case $d=1,2$
\begin{align}
\Gc_i(t,s) u_0
    &:= a_{i}\(s,\Mc_x(t,s),\Mc_y(t,s)\) (\d_x^2-\d_x) u^\BS(\sig_0) .
\end{align}
For a general LSV model with $d-1$ factors of volatility we have
\begin{align}
\Mc_y(t,s)
    &=(\Mc_{y_1}(t,s),\Mc_{y_2}(t,s),\ldots,\Mc_{y_{d-1}}(t,s)) .
\end{align}
Therefore, using Theorem \ref{thm:dyson} we have
\begin{align}
\frac{u_n(t)}{\d_\sig u^\BS(\sig_0)}
    &=  \frac{\Lc_n(t,T) u_0(t)}{\d_\sig u^\BS(\sig_0)}
    =       \frac{\widetilde{\Lc}_n(t,T) (\d_x^2-\d_x) u^\BS(\sig_0)}{\tau \sig_0 (\d_x^2-\d_x) u^\BS(\sig_0)} ,
            \label{eq:clear}
\end{align}
where
\begin{equation}
\widetilde{\Lc}_n(t,T)
    =  \sum_{k=1}^n
            \int_{t}^T \dd t_1 \int_{t_1}^T \dd t_2 \cdots \int_{t_{k-1}}^T \dd t_k
            \sum_{i \in I_{n,k}}
            \Gc_{i_1}(t,t_1) \cdots
            \Gc_{i_{k-1}}(t,t_{k-1})
            a_{i_k}\(s,\Mc_x(t,t_k),\Mc_y(t,t_k)\) . \label{eq:L.tilde}
\end{equation}
It is clear that $\widetilde{\Lc}_n(t,T)$ is a differential operator that takes derivatives with respect to $x$ and $y$ and has coefficients that depend on $(t,T,x,y)$.  Noting that $\d_y^m u^\BS(\sig_0) = 0$ for all $m \geq 1$, it is clear from \eqref{eq:clear} that $u_n/\d_\sig u^\BS(\sig_0)$ is of the form
\begin{align}
\frac{u_n(t)}{\d_\sig u^\BS(\sig_0)}
    &=  \sum_m \chi_m^{(n)}(t,T,x,y) \frac{\d_x^m (\d^2_x - \d_x) u^\BS(\sig_0)}{(\d^2_x - \d_x) u^\BS(\sig_0)} .
    \label{eq:form.2}
\end{align}
Equation \eqref{eq:form} follows from equation \eqref{eq:form.2} and Lemma \ref{lemma:ratio}.  The sequence of coefficients $(\chi_m^{(n)})$ must be computed on a case-by-case basis because the $(\chi_m^{(n)})$ depend on the coefficients of the generator $\Ac(t)$ and the choice of polynomial expansion $(\Ac_n(t))$.
\end{proof}
From Propositions \ref{prop.sig.k} and \ref{prop.un} it is apparent that, as long as the iterated
time integrals in \eqref{eq:L.tilde} can be computed explicitly (which is always the case when the
coefficients in the polynomial expansion $(\Ac_n(t))$ are piece-wise polynomial in
time), every term in \eqref{eq:sig.n_bis} can be computed without the need for numerical
integration or special functions.
\par
Explicit expressions for each $\sig_n$ in the sequence $(\sig_n)_{n \geq 1}$ can be computed by
hand.  However, since the number of terms grows quickly with $n$, it is helpful to use a computer
algebra program such as Wolfram's Mathematica.  In Appendix \ref{sec:impvol.2}, we provide
explicit expressions for $\sig_n$ for $n \leq 2$ the coefficients of $\Ac(t)$ are expanded as a
Taylor series, as in Example \ref{example:Taylor}.  On the authors' website, we also provide
Mathematica notebooks which contains the expressions for $\sig_n$ for $n \leq 4$ for the LSV
models described in Section \ref{sec:examples}.
\begin{remark}
When the risk-free rate of interest is a deterministic function of time $r(t)$, the implied volatility results above hold with $k \to k - \int_t^T r(s) \dd s$.
\end{remark}

\section{Asymptotic error estimates {for Taylor expansions}}
\label{sec:error}
In this section we provide pointwise short-time error estimates for {the} approximate solution of
Cauchy problem \eqref{eq:v.pde} { discussed in Section \ref{sec:model}, as well as for the
approximate implied volatility presented in Section \ref{sec:impvol}}. Throughout this section we
shall assume that $T_{0}>0$ and $N\in\mathbb{N}_{0}$ are fixed and the coefficients of the
operator $\Ac(t)$ in \eqref{operator_AA} satisfy the following assumption:
\begin{assumption}\label{assumption:parabol_holder_bonded}
There exists a positive constant $M$ such that:
\begin{enumerate}
\item[i)] {\it Uniform ellipticity:}
\begin{align}\label{cond:parabolicity}
 M^{-1}|\xi|^2< \sum_{i,j=1}^{d}a_{ij}(t,z)\xi_{i}\xi_{j}< M |\xi|^2,\qquad
t\in\left[0,T_{0}\right],\ z,\xi\in\mathbb{R}^d.
\end{align}
\item[ii)] {\it Regularity and boundedness:} \  the coefficients $a_{ij},a_{i}\in C\left(\left[0,T_{0}\right]\times\R^{d}\right)$ and $a_{ij}(t,\cdot),a_{i}(t,\cdot)\in
C^{N+1}(\R^{d})$, with their partial derivatives of all orders bounded by $M$, uniformly with
respect to $t\in \left[0,T_{0}\right]$.
\end{enumerate}
\end{assumption}
Under Assumption \ref{assumption:parabol_holder_bonded} it is well-known that $\Ac(t)$ admits a
\emph{fundamental solution} $\Gamma(t,z;T,\zeta)$, which is the solution of the Cauchy problem
\eqref{eq:v.pde} with $\varphi=\delta_\zeta$. Equivalently, for any $T\in\left]0,T_{0}\right[$ and for
any measurable function $\varphi$ with at most exponential growth, the backward parabolic Cauchy
problem \eqref{eq:v.pde} admits a {unique} classical solution $u$, which is given by
\begin{align}\label{eq:convolution_u}
u(t,z)=\int_{\mathbb{R}^d} \Gamma(t,z;T,\zeta) \varphi(\zeta) \dd \zeta,\qquad {t\in\left[0,T\right[},\
z\in\mathbb{R}^d.
\end{align}
Furthermore, by the Feynman-Kac representation theorem, the function $\Gamma(t,z;T,\zeta)$ is also the transition density of the stochastic process generated by $\Ac(t)$.

\begin{remark}
\label{rmk:andrea}
Assumption \ref{assumption:parabol_holder_bonded} can be considerably relaxed.  The main results
(Theorem \ref{th:error} and Corollary \ref{cor:errorprice} below) have been recently extended in
\cite{PP_compte_rendu}, to include the majority of popular models
in mathematical finance (e.g. CEV, Heston, SABR, three-halves, etc.). 
\end{remark}

Consider now the Taylor polynomial expansion discussed in
Example \ref{example:Taylor}. It will be helpful to explicitly indicate the dependence on the
{expansion} point $\bar{z}$. In particular, for any {$\bar{z}\in\mathbb{R}^d$}, we
consider the polynomial expansion $(\Ac^{(\bar{z})}_n(t))_{0\leq n\leq N}$, given by
\begin{align}\label{eq:taylor_polyn_bis}
\Ac^{(\bar{z})}_n(t,z)\equiv\Ac^{(\bar{z})}_n(t)
    &:=  \sum_{|\alpha |\leq 2}  a^{(\bar{z})}_{\alpha,n}(t,z) D_z^{\alpha} \qquad a^{(\zbar)}_{\alpha,n}(\cdot,z)
  =        \sum_{|\beta|=n}\frac{D_z^{\beta}a_{\alpha}(\cdot,\bar{z})}{\beta!}(z-\bar{z})^{\beta}, &
n
    &\leq N,
\end{align}
Now, {fix a maturity date $T$}.  We define the $N$-th order Taylor approximation centered at
$\zbar\in\mathbb{R}^d$, of $\Gamma$ and $u$ respectively, as
\begin{align}
\bar{u}^{(\bar{z})}_{N}(t,z)
    &:= \sum_{n=0}^N u^{(\bar{z})}_n(t,z) , &
\bar{\Gamma}^{(\bar{z})}_{N}(t,z,T,\zeta)
    &:= \sum_{n=0}^N \Gamma^{(\bar{z})}_n (t,z,T,\zeta) , &
 t
        &< T, &
z,\zeta
        &\in \mathbb{R}^{d}, \label{eq:u.approx_xbar}
\end{align}
where the functions
\begin{align}
u^{(\bar{z})}_n(t,\cdot)
    &=\Lc^{(\bar{z})}_n(t,T) u^{(\bar{z})}_0(t,\cdot) , &
\Gamma^{(\bar{z})}_n(t,\cdot;T,\zeta)
    &=\Lc^{(\bar{z})}_n(t,T) \Gamma^{(\bar{z})}_0(t,\cdot;T,\zeta) , \label{eq:note.2.ref}
\end{align}
are as given in Theorem \ref{thm:dyson}. {Note that $\bar{u}^{(\bar{z})}_{N}$ is defined for a
fixed $T$, as indicated by \eqref{eq:note.2.ref}.}
Note also that we have once again used the superscript $\bar{z}$  above to emphasize the
dependence on the initial point of the Taylor expansion.
For the particular choice $\zbar=z$, we give the following definition:
\begin{definition}
\label{def:taylor} {For a fixed maturity date $T$,} we define the \emph{$N$-th order Taylor
approximations} of $u$ and $\Gamma$, respectively, as
\begin{align}
\bar{u}_{N}(t,z)
    &:= \bar{u}^{(z)}_{N}(t,z) , &
\bar{\Gamma}_{N}(t,z;T,\zeta)
    &:= \bar{\Gamma}^{(z)}_{N}(t,z;T,\zeta)
                , \label{eq:u.approx}
\end{align}
where $\bar{u}^{(z)}_{N}(t,z)$ and $\bar{\Gamma}^{(z)}_{N}(t,z;T,\zeta)$ are as defined in \eqref{eq:u.approx_xbar}-\eqref{eq:note.2.ref}.
\end{definition}

We now give analogous definitions for the implied volatility expansion. As we did in Section
\ref{sec:impvol}, we use the notation {$(x,y)\in\mathbb{R}\times\mathbb{R}^{d-1}$} to indicate a
point in $\mathbb{R}^d$, {where we separate $x$ from all other components in order to distinguish
the log-price from all the other variables} (e.g. variance process, vol-vol process, etc.). {For a
Call option with maturity date $T$ and log strike $k$}, we define the $N$-th order Taylor
approximation centered at $(\xbar,\bar{y})\in\mathbb{R}\times\mathbb{R}^{d-1}$ of the implied
volatility $\sigma$, as
\begin{align}
 \bar{\sigma}^{(\bar{x},\bar{y})}_N(t,x,y,k)
    &:=\sigma^{(\bar{x},\bar{y})}_0(t) + \sum_{n=0}^{N}\sigma^{(\bar{x},\bar{y})}_n(t,x,y,k), &
t
    &< T, &
(x,y)
    &\in \mathbb{R}\times\mathbb{R}^{d-1} , \label{eq:sig.approx_xbar}
\end{align}
where, for sake of clarity we recall
\begin{align}
\sig^{(\bar{x},\bar{y})}_0(t)
    &=\sqrt{ \frac{2}{T-t} \int_t^T   a_{(2,0,\ldots,0)}(s,\bar{x},\bar{y}) \, \dd s},\label{eq:sigma_0_tris} \\
\sig^{(\bar{x},\bar{y})}_n(t,x,y,k)
    &=  \frac{u^{(\bar{x},\bar{y})}_n(t,x,y,k)}{\d_\sig u^\BS\big(\sig^{(\bar{x},\bar{y})}_0(t);T-t,x,k\big)}
            - \frac{1}{n!} \sum_{h=2}^{n}
                \mathbf{B}_{n,h}\(1!\,\sig^{(\bar{x},\bar{y})}_{1},2!\,\sig^{(\bar{x},\bar{y})}_{2},\dots,(n-h+1)!\, \sig^{(\bar{x},\bar{y})}_{n-h+1}\)
        \\ &\qquad \times
            \frac{\d_\sig^h u^\BS\big(\sig^{(\bar{x},\bar{y})}_0(t);T-t,x,k\big)}{\d_\sig u^\BS\big(\sig^{(\bar{x},\bar{y})}_0(t);T-t,x,k\big)},
 \qquad\qquad n \geq 1 \label{eq:sigma_n_tris} , \\
u^{(\bar{x},\bar{y})}_n(t,x,y,k)
    &=  \Lc_n^{(\bar{x},\bar{y})}(t,T) u^{(\bar{x},\bar{y})}_0(t,x,k)
    =       \Lc_n^{(\bar{x},\bar{y})}(t,T) u^\BS(\sig_0^{(\bar{x},\bar{y})}(t);T-t,x,k) . \label{eq:matt2}
\end{align}
A few notes are in order.  {First, we have added the argument $k$ to the function
$u_n^{(\bar{x},\bar{y})}$ to indicate its dependence on $\log$ strike. Second, the function
$u_n^{(\bar{x},\bar{y})}$ depends on the maturity date $T$, as indicated by \eqref{eq:matt2}.}
{Third, each $\sig^{(\bar{x},\bar{y})}_{n}$ in the sequence $(\sig^{(\bar{x},\bar{y})}_{n})_{n
\geq 1}$ depends on $(t,x,y,k)$.  Though, for clarity, we have not written all of these arguments
in
$\mathbf{B}_{n,h}\(1!\,\sig^{(\bar{x},\bar{y})}_{1},2!\,\sig^{(\bar{x},\bar{y})}_{2},\dots,(n-h+1)!\,
\sig^{(\bar{x},\bar{y})}_{n-h+1}\)$.}
Fourth, we have once again {explicitly indicated with a superscript $(\bar{x},\bar{y})$ the
dependence on the initial point of the Taylor expansion}. For the particular choice $\xbar=x$ and
$\bar{y}=y$, we make the following definition:
\begin{definition}
\label{def:taylor.sigma} {For a Call option with $\log$ strike $k$ and maturity $T$,} we define
the  \emph{$N$-th order Taylor approximation of the implied volatility} $\sigma$ as
\begin{align}
\bar{\sigma}_N(t,x,y,k):=\bar{\sigma}^{({x,y})}_N(t,x,y,k), \label{eq:sig.approx}
\end{align}
where $\bar{\sigma}^{({x,y})}_N(t,x,y,k)$ is as defined in \eqref{eq:sig.approx_xbar}-\eqref{eq:sigma_0_tris}-\eqref{eq:sigma_n_tris}-\eqref{eq:matt2}.
\end{definition}

\subsection{Error estimates for the transition density and prices}
\label{sec:price.error}
The following theorem provides an asymptotic pointwise estimate as $t\to T^{-}$ for the error
introduced by replacing the exact transition density $\Gamma$ with the $N$-th order approximation
$\bar{\Gamma}_{N}$.
\begin{theorem}\label{th:error}
Let Assumption \ref{assumption:parabol_holder_bonded} hold and let $0<T\leq T_{0}$.
Then, for any 
$\varepsilon>0$ we have
\begin{align}\label{eq:error_estimate}
 \left| \Gamma(t,z;T,\zeta)-\bar{\Gamma}_{N}(t,z;T,\zeta) \right| \leq C \, (T-t)^{\frac{N+1}{2}}\Gamma^{M+\varepsilon} (t,z;T,\zeta), \qquad 0\leq t<T,\
 z,\zeta\in\mathbb{R}^d,
\end{align}
where {$\bar{\Gamma}_{N}(t,z;T,\zeta)$ is as defined in \eqref{eq:u.approx},}
$\Gamma^{M+\varepsilon}(t,z;T,\zeta)$ is the fundamental solution of the
heat operator
\begin{align}\label{eq:heatoperator}
H^{M+\varepsilon}=(M+\varepsilon) \sum_{i=1}^d \partial^2_{z_i}+\partial_t,
\end{align}
and $C$ is a positive constant that depends only on $M,N,T_{0}$ and $\varepsilon$.
\end{theorem}
\noindent
Combining Theorem \ref{th:error} with \eqref{eq:convolution_u} we obtain an asymptotic estimate for $|u(t,z)-\bar{u}_{N}(t,z)|$, the pricing error.
\begin{corollary}\label{cor:errorprice}
Under the assumptions of Theorem \ref{th:error}, for any $0<T\leq T_{0},\varepsilon>0$ we have
 \begin{align}\label{eq:error_estimate_price}
 \left|u(t,z)-\bar{u}_{N}(t,z)\right|\leq C \, (T-t)^{\frac{N+1}{2}}\int_{\mathbb{R}^d}
 \Gamma^{M+\varepsilon}(t,z;T,\zeta) \varphi(\zeta) \dd \zeta, \qquad 0\leq t<T,\
 z\in\mathbb{R}^d.
\end{align}
{where $\bar{u}_{N}(t,z)$ is as defined in \eqref{eq:u.approx}.}
\end{corollary}
The proof of Theorem \ref{th:error} relies on the following Gaussian estimates (see
\cite{friedman-parabolic}, Chapter 1).
\begin{lemma}\label{lem:gaussian_estimates}
Let $\Ac(t)$ be a differential operator satisfying Assumption \ref{assumption:parabol_holder_bonded} and let $\Gamma=\Gamma(t,z;T,\zeta)$ be the fundamental solution corresponding to $\Ac(t)$. Then, for any $\varepsilon>0$ and $\b,\g\in\mathbb{N}_{0}^{d}$ with $|\g|\le N+3$, we have
\begin{align}\label{Gaua}
 |(z-\zeta)^{\beta}\, D_z^{\g}\Gamma(t,z;T,\zeta)|\le C \, (T-t)^{\frac{|\b|-|\g|}{2}}\Gamma^{M+\varepsilon}(t,z;T,\zeta),\qquad 0\leq t<T \leq T_{0},\quad z,\zeta\in\R^{d},
\end{align}
where $\Gamma^{M+\varepsilon}$ is the fundamental solution of the heat operator
\eqref{eq:heatoperator} and $C$ is a positive constant, which depends only on
$M,N,T_{0},\varepsilon$ and $|\b|$.
\end{lemma}
We also need the following preliminary estimates (see \cite[Lemma 6.23]{LPP4})
\begin{lemma}\label{lem:eq:der_u_n}
Under the assumptions of Theorem \ref{th:error}, for any $n\in\mathbb{N}$ with $n\leq N$, $\epsilon>0$, and for any $\beta
\in\mathbb{N}_0^{d}$, we have
\begin{align}
\label{eq:der_u_n}
 \left|D^{\beta}_{z}\Gamma^{(\bar{z})}_n(t,z)\right|\leq C \, (T-t)^{\frac{n-|\beta|}{2}}\left(1+|z-\bar{z}|^{n}\left(T-t\right)^{-\frac{n}{2}}
 \right)\Gamma^{M+\varepsilon}(t,z;T,\zeta),
\end{align}
which holds for $0\leq t<T\leq T_{0},\ z,\zeta,\bar{z}\in\mathbb{R}^d$. Here, the function
$\Gamma^{M+\varepsilon}$ is the fundamental solution of the heat operator \eqref{eq:heatoperator}
and $C$ is a positive constant, which depends only on $M,N,T_{0},\varepsilon$ and $|\b|$.
\end{lemma}
\begin{proof}[Proof of Theorem \ref{th:error}]
From \cite[Theorem 3.8]{LPP4}, for any given $T\leq T_{0}$, the functions
$(u^{(\bar{z})}_n)_{n \geq 1}$ given by \eqref{eq:un}-\eqref{eq:Ln} can be equivalently
defined as the unique non-rapidly increasing solutions of the following sequence of nested
heat-type Cauchy problems:
\begin{align}\label{eq:v.n.pide.xbar}
  \begin{cases}
 \left(\d_t + \Ac^{(\bar{z})}_0(t)\right) u^{(\bar{z})}_n(t,z) =  - \sum\limits_{h=1}^{n} \Ac^{(\bar{z})}_h(t) u^{(\bar{z})}_{n-h}(t,z),\qquad &  t<T,\ z\in\mathbb{R}^d, \\
  u^{(\bar{z})}_n(T,z) = 0, & z \in\mathbb{R}^d.
  \end{cases}
\end{align}
The thesis then follows directly from \cite[Theorem 3.10]{LPP4}.  For completeness, we provide here a sketch of the proof given in \cite{LPP4}. By \eqref{eq:v.n.pide.xbar} it is easy to prove that
{$v^{(\bar{z})}:=u-\bar{u}^{(\bar{z})}_{N}$} solves
\begin{align}\label{eq:v.n.pide.xbar2}
\begin{cases}
 (\d_t + \Ac(t) ) v^{(\bar{z})}(t,z) =- \sum\limits_{n=0}^N (\Ac(t) - \bar{\Ac}^{(\bar{z})}_n(t)) u^{(\bar{z})}_{N-n}(t,z),
 \qquad & t<T,\ z\in\mathbb{R}^d, \\
 v^{(\bar{z})}(T,z) =  0, &  z \in\mathbb{R}^d,
\end{cases}
\end{align}
where we have defined
\begin{align}\label{eq:A_bar_n}
\bar{\Ac}^{(\bar{z})}_n(t) = \sum_{i=0}^n  \Ac^{(\bar{z})}_i(t).
\end{align}
Thus, by Duhamel's principle we obtain
\begin{align}
 u(t,z)-\bar{u}_{N}(t,z) = \int_t^T \int\limits_{\mathbb{R}^d} \Gamma(t,z;s,\xi) \sum_{n=0}^N
 \big(\Ac(s) -  \bar{\Ac}^{({z})}_n(s)\big) u^{({z})}_{N-n}(s,\xi)\, \dd \xi \dd s, \quad
 t<T,\ z\in\mathbb{R}^d.
\end{align}
Now, by \eqref{eq:taylor_polyn_bis} we have
\begin{align}
|(\Ac(s)-\bar{\Ac}^{({z})}_n(s))u^{({z})}_{N-n}(s,\xi)|
    &\leq   \sum_{|\alpha |\leq 2}  \Big| a^{({z})}_{\alpha}(s,\xi)-\sum_{i=0}^{n}a^{({z})}_{\alpha,n}(s,\xi)\Big| |D_{\xi}^{\alpha}u^{({z})}_{N-n}(s,\xi)| \\
&=   \sum_{|\alpha |\leq 2}  \Big| a_{\alpha}(s,\xi)-\sum_{i=0}^{n}\sum_{|\b|=n}\frac{D_z^{\b}a_{\alpha}(s,z)}{\b!}(\xi-z)^{\b}
\Big| \big|D_{\xi}^{\alpha}u^{({z})}_{N-n}(s,\xi) \big|
\\
&\leq M |\xi-z|^{n+1}  \sum_{|\alpha |\leq 2}  \big| D_{\xi}^{\alpha}u^{({z})}_{N-n}(s,\xi) \big|,
\end{align}
where the last line follows by the hypothesis (ii) in Assumption \ref{assumption:parabol_holder_bonded} on the coefficients $(a_{\alpha})_{|\alpha|\leq 2}$. Finally, by considering $u^{({z})}_n(t,z)=\Gamma^{({z})}_n(t,z,;T,\zeta)=\Lc^{({z})}_n(t,T)\Gamma^{({z})}_0(t,z;T,\zeta)$, we obtain
\begin{align}
 |\Gamma(t,z;T,\zeta)-\bar{\Gamma}_{N}(t,z;T,\zeta)| \leq M \sum_{n=0}^N \sum_{|\alpha |\leq 2} \int_t^T \int\limits_{\mathbb{R}^d} \Gamma(t,z;s,\xi)
|\xi-z|^{n+1}|D_{\xi}^{\alpha} \Gamma^{({z})}_{N-n}(s,\xi;T,\zeta)|\, \dd \xi \dd s.
\end{align}
The thesis now follows by repeatedly applying the Gaussian estimates \eqref{Gaua} and \eqref{eq:der_u_n}, along with the semigroup property
\begin{align}
\int_{\mathbb{R}^d}  \Gamma^{M+\varepsilon}(t,z;s,\xi) \Gamma^{M+\varepsilon}(s,\xi;T,\zeta) \dd \xi \dd s =  \Gamma^{M+\varepsilon}(t,z;T,\zeta)\qquad t<s<T,\quad z,\zeta\in\mathbb{R}^d.
\end{align}
\end{proof}

\subsection{Short-time asymptotics for the implied volatility}
\label{sec:error.impvol} We provide error estimates for the $N$-th order implied volatility approximation
$\bar{\sigma}_N$, defined in \eqref{eq:sig.approx},
on the subset $|x-k|\leq\lam\sqrt{T-t}$ {where $\lam$ is an arbitrary, but fixed, positive
constant}.
\begin{theorem}\label{th:IV_error}
Let Assumption \ref{assumption:parabol_holder_bonded} hold and let {$\lam>0$}. 
{Denote by $\sigma(t,x,y,k)$ the exact implied volatility of a Call option, with $\log$ strike $k$
and maturity $T$.} {That is, $\sigma(t,x,y,k)$ is the unique positive solution of
$u^\BS(\sig;T-t,x,k)=u(t,x,y,k)$, where $u$ is the classical solution $\sig$ of \eqref{eq:v.pde}
with time $T$ terminal conditions $\varphi(x) = (\ee^x - \ee^k)^+$.} 
Then the $N$-th order implied volatility approximation $\bar{\sigma}_{N}(t,x,y,k)$, defined in
\eqref{eq:sig.approx}, satisfies
\begin{equation}\label{eq:IV_estimate}
 \left| \sigma(t,x,y,k)-\bar{\sigma}_{N}(t,x,y,k) \right| \leq C (T-t)^{\frac{N+1}{2}}, \qquad 0\leq t<T\le T_{0},\ y\in\mathbb{R}^{d-1},
\ |x-k|\leq \lam\sqrt{T-t},
\end{equation}
where $C$ is a positive constant that depends only on $M,N,T_{0}$ and $\lam$.
\end{theorem}



\begin{remark}
In the particular case $d=1$, the above result is consistent with \cite[Theorem
22]{BompisGobet2012} where an implied volatility approximation for local volatility models has
been derived. A direct computation shows that such an expansion is equivalent to our
$\bar{\sigma}_2$.
Although Theorem \ref{th:IV_error} holds true for any order $N \in\mathbb{N}_0$, and any dimension
$d\in\mathbb{N}$, the estimate in \cite{BompisGobet2012} was proved by the authors under milder
assumptions for the generator $\Ac(t)$, and for three different choices of the initial point
$\bar{x}$ of the Taylor expansion: $\bar{x}=x$, $\bar{x}=k$ and $\bar{x}=\frac{x+k}{2}$.
\end{remark}
\begin{remark}
Theorem \ref{th:IV_error}  also provides us with an explicit representation for the $n$-th order
derivative with respect to $T$, of the implied volatility surface at $x=k$ and $T=t$.  More
precisely, as a corollary of \eqref{eq:IV_estimate} we have:
\begin{equation}\label{eq:T_der_IV}
\partial^n_t \sigma(t,x,y,k)|_{t=T,k=x}= \partial^n_t \bar{\sigma}_{N}(t,x,y,k)|_{t=T,k=x}, \qquad \forall N\geq 2n.
\end{equation}
A direct computation shows that, for $n=0$, the representation \eqref{eq:T_der_IV} is consistent
with the well-known results by \cite{berestycki2002asymptotics} and
\cite{berestycki-busca-florent}. 
It is also easy to check that our expansion gives the correct slope of the implied volatility at the money in the limit as $t \to T$.
For the special case $d=1$, we recover the practitioners' \emph{$1/2$ slope rule}, which
gives the at-the-money slope of implied volatility as one half the slope of the local
volatility function.
\end{remark}

In what follows, the maturity date $T\in(0, T_{0}]$ is fixed.  We recall the Black-Scholes
price
$$u^{\BS}(\s)=u^{\BS}(\sigma;T-t,x,k),$$
as it is in Definition \ref{def:BS} and we denote
by {$(u^{\BS})^{-1}(u;T-t,x,k)=(u^{\BS})^{-1}(u)$} its inverse with respect to the $\sigma$
variable.
We also introduce the following function:
\begin{align}
u(\delta)&=u(\delta;t,x,y,k) :=\sum_{n=0}^N \delta^n u^{(x,y)}_{n}(t,x,y,k) + \delta^{N+1}
\big(u-\bar{u}_N\big)(t,x,y,k) \\ &=
u^{\BS}\big(\s^{(x,y)}_{0}(t);T-t,x,k\big)+\sum_{n=1}^N \delta^n u^{(x,y)}_{n}(t,x,y,k) +
\delta^{N+1} \big(u-\bar{u}_N\big)(t,x,y,k),\qquad \delta\in [0,1], \\
&\label{eq:def_u_delta}
\end{align}
where we have used Remark \ref{rmk:u0=uBS}.  {Note that the function $u^{(x,y)}_{n}(t,x,y,k)$ is
defined for a fixed maturity date $T$, as indicated by \eqref{eq:matt2}.}

\begin{remark}\label{eq:error_estimate_price_call}
It is possible to prove (see \cite{LPP4}) that, in the case of a Call option, the estimate
\eqref{eq:error_estimate_price}, as well as \eqref{eq:der_u_n}, can be improved by exploiting the
local {Lipschitz continuity} of the payoff $\varphi(x)=(e^x-e^k)^+$. More precisely, it is
possible to prove that, for any $\log$-strike $k \in \Rb$, we have
\begin{align}
\left|u(t,x,y,k)-\bar{u}_{N}(t,x,y,k)\right| \leq C \,
(T-t)^{\frac{N+2}{2}}u^{\BS}\big(\sqrt{2M};T-t,x,k\big), \qquad 0\leq t<T,\
 (x,y)\in\mathbb{R}\times\mathbb{R}^{d-1}, 
\end{align}
and that, for any $n\in\mathbb{N}$ with $n\leq N$, we also have
\begin{align}
\label{eq:der_u_n_call}
 \big|u^{({x},y)}_n(t,x,y,k)\big| \leq C \, (T-t)^{\frac{n+1}{2}} u^{\BS}\big(\sqrt{2M};{T-t},x,k\big),\qquad 0\leq t<T,\
 (x,y)\in\mathbb{R}\times\mathbb{R}^{d-1},\ k\in\mathbb{R},
\end{align}
where, {as in Theorem \ref{th:IV_error},} $C$ is a positive constant that only depends on $M,N$
and $T_{0}$.
\end{remark}
The proof of Theorem \ref{th:IV_error} is based on the previous remark and some asymptotic
estimates {of the Black-Scholes price} for short-maturities, which are proved in Appendix
\ref{sec:lemmas}.
\begin{lemma}\label{lem:imp_vol_u_delta}
Let $u(\del)$ be as in \eqref{eq:def_u_delta}. Under the assumptions of Theorem
\ref{th:IV_error}, 
there exists $\t_{0}>0$, only dependent on $M,N,T_{0}$ and $\lam$, such that
\begin{align}
u^{\BS}\big( \sqrt{{2}/{{M}}};{T-t},x,k\big) \leq u(\del)
     \leq u^{\BS}\big(\sqrt{2M};{T-t},x,k\big) , \label{eq:bounds}
\end{align}
or equivalently
\begin{align}\label{eq:volimp_u_d}
 \sqrt{{2}/{M}}\leq \big(u^{\BS}\big)^{-1}(u(\delta);{T-t},x,k)\leq \sqrt{2M}, 
\end{align}
for any {$t\in[T-\t_{0},T)$},  $|x-k|\leq \lam \sqrt{T-t}$, $y\in\mathbb{R}^{d-1}$ and
$\delta\in[0,1]$.
\end{lemma}
\begin{proof}
Throughout this proof $C$ will always denote a positive constant that depends only on $M,N,T_{0}$
and $\lam$. By Remark \ref{eq:error_estimate_price_call}, and since $\delta\in[0,1]$, we obtain
\begin{align}
 \Big|\sum_{n=1}^N \delta^n u^{(x,y)}_{n}(t,x,y,k) + \delta^{N+1} \big(u-\bar{u}_N\big)(t,x,y,k)\Big|& \leq C \, (T-t)
 u^{\BS}\big(\sqrt{2M};{T-t},x,k\big)
 \intertext{(using Lemma \ref{landre1} and Assumption
\ref{assumption:parabol_holder_bonded})}
 & \leq C \, (T-t)
u^{\BS}\big(\sigma^{(x,y)}_0(t);{T-t},x,k\big)\label{eq:inequality_module_reminder_u_delta}
\end{align}
for any $0\leq t<T$, $y\in \mathbb{R}^{d-1}$ and $|x-k| \leq \lam \sqrt{T-t}$.
{Combining \eqref{eq:def_u_delta} and \eqref{eq:inequality_module_reminder_u_delta}, we obtain
\begin{align}
 u(\delta) &\geq (1-C(T-t)) u^{\BS}\big(\sigma^{(x,y)}_0(t);{T-t},x,k\big) .
\end{align}
The lower bound for $u(\del)$ in \eqref{eq:bounds} now follows from inequality \eqref{andre7} in Lemma \ref{landre2}.
To establish the upper bound for $u(\del)$, we combine \eqref{eq:def_u_delta} with \eqref{eq:inequality_module_reminder_u_delta} to obtain
\begin{align}
u(\del)
    \leq (1+C(T-t)) u^{\BS}\big(\sigma^{(x,y)}_0(t);{T-t},x,k\big) .
\end{align}
The upper bound in \eqref{eq:bounds} now follows from inequality \eqref{eq:second} in Lemma
\ref{landre2}. }
\end{proof}
\begin{lemma}\label{lem:estimates_deriv_inv_BS}
Under the assumptions of Theorem \ref{th:IV_error}, for any {$N\in\mathbb{N}$} there exist
positive constants $C$ and $\t_{0}$, only dependent on $M,N,T_{0}$ and $\lam$, such that
\begin{align}\label{eq:deriv_inv_BS_estimate}
 \left| \partial^n_{u} \big(u^{\BS}\big)^{-1}\big(u(\delta;t,x,y,k);{T-t},x,k\big) \right| \leq C \big( e^k \sqrt{T-t}\,\big)^{-n},
\end{align}
for any {$n\le \mathbb{N}$, $t\in[T-\t_{0},T)$}, $|x-k|\leq \lam \sqrt{T-t}$,
$y\in\mathbb{R}^{d-1}$ and $\delta\in[0,1]$.
\end{lemma}
\begin{proof}
Throughout this proof, $C$ will always denote a positive constant only dependent on $M,N,T_{0}$
and $\lam$. Note that, for any $\sigma>0$ we have
\begin{align}
\partial_{\sigma}u^{\BS}(\sigma)
\equiv \partial_{\sigma}u^{\BS}(\sigma;T-t,x,k)
=\frac{e^k \sqrt{T-t} }{\sqrt{2 \pi
  }}\exp\bigg({-\frac{\left(\sigma ^2 (T-t)-2 (x-k)\right)^2}{8 \sigma ^2 (T-t)}}\bigg),
\end{align}
and thus
\begin{align}
 {
\frac{e^k \sqrt{T-t} }{\sqrt{2 \pi
  }}\exp\Big( - \frac{\sigma^2 T_{0}}{8}-\frac{\lam^2}{2\s^2}-\frac{\lam\sqrt{T_{0}}}{2} \Big) }
    &\leq   \partial_{\sigma}u^{\BS}(\sigma
  )\leq \frac{e^k \sqrt{T-t} }{\sqrt{2 \pi
  }}, &
    0
    &\leq t<T, \, |x-k|\leq \lam\sqrt{T-t} .
\end{align}
Therefore, by Lemma \ref{lem:imp_vol_u_delta}, there exists a positive $\t_{0}$, only dependent on
$M,N,T_{0}$ and $\lam$, such that
\begin{align}\label{eq:estimate_below_above_vega_u_delta}
 C \frac{e^k \sqrt{T-t} }{\sqrt{2 \pi
  }} &\leq \partial_{\sigma} u^{\BS}\Big(\big(u^{\BS}\big)^{-1}(u(\delta))
  \Big) \leq \frac{e^k \sqrt{T-t} }{\sqrt{2 \pi}},
\end{align}
for any $y\in\mathbb{R}^{d-1}$, {$t\in[T-\t_{0},T)$},  $|x-k|\leq \lam \sqrt{T-t}$ and
$\delta\in[0,1]$, where $C$ is the positive constant
 $$C= \min_{\sig \in [\sqrt{2/M},\sqrt{2M}]} \exp\Big(  -\frac{\sigma^2 T_{0}}{8}-\frac{\lam^2}{2\s^2}-\frac{\lam\sqrt{T_{0}}}{2} \Big).$$
Furthermore, by combining the second inequality in \eqref{eq:estimate_below_above_vega_u_delta}
with Proposition \ref{prop.sig.k}, we also obtain
\begin{align}\label{eq:estimate_above_vega_n_u_delta}
\Big|  \partial^n_{\sigma} u^{\BS}\Big(\big(u^{\BS}\big)^{-1}(u(\delta))
  \Big) \Big| \leq C e^k \sqrt{T-t}  
  .
\end{align}
We are now prove the thesis by induction on $n$. The case $n=1$ clearly follows from the first
inequality in \eqref{eq:estimate_below_above_vega_u_delta}.  We have
\begin{align}
 \Big|\partial_{u} \big(u^{\BS}\big)^{-1}(u(\delta))\Big|=\frac{1}{\partial_{\sigma}u^{\BS}\big(\big(u^{\BS}\big)^{-1}(u(\delta))
 \big)}\leq  \frac{C e^{-k}}{ \sqrt{T-t}}.
\end{align}
Let us now assume \eqref{eq:deriv_inv_BS_estimate} holds true for any 
$m\leq n$, and prove it holds true for $n+1$. By Fa\`a di Bruno's formula (see Appendix
\ref{append:faa_bell}, Eq. \eqref{eq:Faa_di_Bruno_appendix}), we have
\begin{align}
\partial^{n+1}_{u} \big(u^{\BS}\big)^{-1}(u
)=\frac{\sum_{h=2}^{n+1} \partial^h_{\sigma}u^{\BS}\big(\big(u^{\BS}\big)^{-1}(u
)
\big) \mathbf{B}_{n+1,h}\Big(
\partial_{u} \big(u^{\BS}\big)^{-1}(u
),
\cdots,\partial^{n-h+2}_{u} \big(u^{\BS}\big)^{-1}(u
)
  \Big)}{\partial_{\sigma}u^{\BS}\big(\big(u^{\BS}\big)^{-1}(u
)
\big)},
\end{align}
and thus, by \eqref{eq:estimate_below_above_vega_u_delta} and \eqref{eq:estimate_above_vega_n_u_delta}, we obtain
\begin{align}
&\big|\partial^{n+1}_{u} \big(u^{\BS}\big)^{-1}(u(\delta)
)\big|
    \leq C \sum_{h=2}^{n+1}
\Big|\mathbf{B}_{n+1,h}\Big(
\partial_{u} \big(u^{\BS}\big)^{-1}(u(\delta)
),
\cdots,\partial^{n-h+2}_{u} \big(u^{\BS}\big)^{-1}(u(\delta)
)
  \Big)\Big| \\
  &
\leq C \sum_{h=2}^{n+1}\ \sum_{j_1,\cdots,j_{n-h+2}}
\big|
\partial_{u} \big(u^{\BS}\big)^{-1}(u(\delta)
)\big|^{j_1} \cdots
\big|\partial^{n-h+2}_{u} \big(u^{\BS}\big)^{-1}(u(\delta)
)
  \big|^{j_{n-h+2}} &
    &\text{(by \eqref{eq:Bell_polyn_appendix} in Appendix \ref{append:faa_bell})}\\[-2em]
& \leq C \sum_{h=2}^{n+1}\ \sum_{j_1,\cdots,j_{n-h+2}}
\big(e^k\sqrt{T-t}\,\big)^{-j_1} \cdots
\big(e^k\sqrt{T-t}\,\big)^{-j_{n-h+2}}  &
&\text{(by inductive hypothesis)} \\[-2em]
& \leq C \sum_{h=2}^{n+1}\ \sum_{j_1,\cdots,j_{n-h+2}}
\big(e^k\sqrt{T-t}\,\big)^{-(n+1)}=  C\big(e^k\sqrt{T-t}\,\big)^{-(n+1)},
\end{align}
where the last inequality follows from the second identity of \eqref{eq:relation_indexes_bell} in Appendix \ref{append:faa_bell}.
This concludes the proof.
\end{proof}

\begin{proof}[Proof of Theorem \ref{th:IV_error}]
Throughout this proof $C$ will indicate a positive constant only dependent on $M,N,T_{0}$ and
$\lam$. It suffices to prove the thesis \eqref{eq:IV_estimate} for small $T-t$. We start by
recalling the function $g(\delta)$, which has already been used in Section
\ref{subsec:imp_vol_formal} to carry out the formal expansion of the implied volatility, i.e.
\begin{align}
g(\delta)=g(\delta;t,x,y,k):=\big(u^{\BS}\big)^{-1}\big(u(\delta;t,x,y,k);{T-t},x,k\big),\qquad
\delta\in [0,1].
\end{align}
By definition of $u(\delta)$ in \eqref{eq:def_u_delta}, it is clear that
\begin{align}\label{eq:sigma_g1}
\sigma(t,x,y,k)=g(1;t,x,y,k).
\end{align}
Furthermore, {with $\bar{\sigma}_{N}(t,x,y,k)$ as defined in \eqref{eq:sig.approx}}, we have
\begin{align}\label{eq:sigmabar_g0tay}
\bar{\sigma}_{N}(t,x,y,k)=\sigma^{(x,y)}_0(t) + \sum_{n=0}^{N}\sigma^{(x,y)}_n(t,x,y,k)=\sum_{n=0}^N \frac{1}{n!}\partial^n_{\delta}g(\delta;t,x,y,k)\big|_{\delta=0},
\end{align}
since, by \eqref{eq:def_u_delta} and \eqref{ae6} we have, respectively,
$g(\delta)|_{\delta=0}=\sigma^{(x,y)}_0,$
and
$\partial^n_{\delta}g(\delta)\big|_{\delta=0}=\sigma^{(x,y)}_n$
for
$1\leq n\leq N$.
Now, by \eqref{eq:sigma_g1}-\eqref{eq:sigmabar_g0tay}, and by the Taylor theorem with Lagrange remainder, there exist $\bar{\delta}\in [0,1]$ such that
\begin{align}
\sigma-\bar{\sigma}_{N}&=g(1)-\sum_{n=0}^N \frac{1}{n!}\partial^n_{\delta}g(0)=\frac{1}{(N+1)!}\partial^{N+1}_{\delta}g(\bar{\delta}) \\
&=\frac{1}{(N+1)!}\sum_{h=1}^{N+1} \partial^h_{u} \big(u^{\BS}\big)^{-1}\big( u(\bar{\delta})
\big) \mathbf{B}_{N+1,h}\left(  \partial_{\delta}u(\bar{\delta}),\partial^2_{\delta}u(\bar{\delta}),
\cdots,\partial_{\delta}^{N-h+2}u(\bar{\delta})  \right) , \label{eq:diff_sigma_sigma_bar}
\end{align}
by \eqref{eq:Faa_di_Bruno_appendix} in Appendix \ref{append:faa_bell}.
Now, 
by \eqref{eq:def_u_delta} and Remark \ref{eq:error_estimate_price_call}, we obtain
\begin{align}\label{eq:maybelast}
|\partial^n_{\delta}u(\bar{\delta})|
    &\leq C \bigg(\sum_{h=n}^{N} |u^{(x,y)}_n| +  |u-\bar{u}_N| \bigg)
    \leq C (T-t)^{\frac{n+1}{2}} u^{\BS}\big(\sqrt{2M};{T-t},x,k\big). 
\end{align}
Therefore, for any $1\leq h\leq N+1$, by \eqref{eq:Bell_polyn_appendix} in Appendix
\ref{append:faa_bell} we have
\begin{align}
\big|\mathbf{B}_{N+1,h}\big(
\partial_{\delta}u(\bar{\delta}),\partial^2_{\delta}u(\bar{\delta}),\cdots,\partial_{\delta}^{N-h+2}u(\bar{\delta})\big)\big|
 &\leq C  \sum_{j_1,\cdots,j_{N-h+2}} \big|
\partial_{\delta}u(\bar{\delta})
\big|^{j_1}\big|
\partial^2_{\delta}u(\bar{\delta})
\big|^{j_2} \cdots
\big|\partial^{N-h+2}_{\delta} u(\bar{\delta})
  \big|^{j_{N-h+2}} \\
&\leq C  
\big(\sqrt{T-t}\,\big)^{N+h+1}\left(u^{\BS}\big(\sqrt{2M};{T-t},x,k\big)\right)^{h}.\label{eq:very_final}
\end{align}
where in the last inequality we have used \eqref{eq:maybelast} and both the identities from
\eqref{eq:relation_indexes_bell} in Appendix \ref{append:faa_bell}.
Combining \eqref{eq:deriv_inv_BS_estimate} and \eqref{eq:very_final} with
\eqref{eq:diff_sigma_sigma_bar}, we obtain
  $$\left|\sigma-\bar{\sigma}_{N}\right|\le C(T-t)^{\frac{N+1}{2}}\sum_{h=1}^{N+1}\left(e^{-k} u^{\BS}\big(\sqrt{2M};{T-t},x,k\big)\right)^{h}.$$
The thesis finally follows since $e^{-k} u^{\BS}\big(\sqrt{2M};{T-t},x,k\big)\le
e^{\lam\sqrt{T-t}}$ for $|x-k|\le \lam\sqrt{T-t}$.
\end{proof}

%
%

\section{Implied volatility examples}
\label{sec:examples} In this section we use the results of Section \ref{sec:impvol} to compute
approximate model-induced implied volatilities
under
four
different model dynamics in which European option prices can be computed explicitly.
\begin{itemize}[noitemsep]
\item Section \ref{sec:CEV}: CEV local volatility model
\item Section \ref{sec:Heston}: Heston stochastic volatility model
\item Section \ref{sec:Three-Halves}: 3/2 stochastic volatility model
\item Section \ref{sec:SABR}: SABR local-stochastic volatility model
\end{itemize}
We note that all of the above models fail to satisfy the rigorous assumptions required in Theorems
{\ref{th:error} and \ref{th:IV_error}} to prove the error bounds \eqref{eq:error_estimate}. {
However, as mentioned in Remark \ref{rmk:andrea} Theorem \ref{th:error} and Corollary
\ref{cor:errorprice} have been recently extended in \cite{PP_compte_rendu}, to include all of the
examples presented here. 
}
\par
In
three
of the
four
examples that follow we use a Taylor series polynomial expansion of $\Ac(t)$ as in Example \ref{example:Taylor}.  In these three cases, approximate implied volatilities can be computed using the formulas given in Appendix \ref{sec:impvol.2}.  For the Heston model, we use the time-dependent Taylor expansion of $\Ac(t)$ as in Example \ref{example:TimeTaylor}.  In all cases, Mathematica notebooks containing the implied volatility formulas are available free of charge on the authors' website.


\subsection{CEV local volatility model}
\label{sec:CEV}
In the Constant Elasticity of Variance (CEV) local volatility model of \cite{CoxCEV}, the dynamics of the underlying $S$ are given by
\begin{align}
\dd S_t
    &=  \del S_t^{\beta-1} S_t \dd W_t , &
S_0
    &=  s > 0 .
\end{align}
The parameter $\beta$ controls the relationship between volatility and price.  When $\beta < 1$,
volatility increases as $S \to 0^+$. This feature, referred to as the \emph{leverage effect}, is
commonly observed in equity markets.  When $\beta<1$, one also observes a negative at-the-money
skew in the model-induced implied volatility surface.  Like the leverage effect, a negative
at-the-money skew is commonly observed in equity options markets. The origin is attainable
when $\beta<1$. In order to prevent the process $S$ from taking negative values, one typically
specifies zero as an absorbing boundary.  Hence, the state space of $S$ is $[0,\infty)$. In log
notation $X := \log S$, we have the following dynamics \footnote{Here 
we define $\log 0 := \lim_{x \searrow 0} \log x = - \infty$.}
\begin{align}
\dd X_t
    &=  -\frac{1}{2} \del^2 \ee^{2(\beta-1)X_t} \dd t + \del \, \ee^{(\beta-1)X_t} \dd W_t , &
X_0
    &=  x := \log s . \label{eq:model.CEV}
\end{align}
The generator of $X$ is given by
\begin{align}
\Ac
    &=  \frac{1}{2} \del^2 \ee^{2(\beta-1)x} ( \d_x^2 - \d_x ).
\end{align}
Thus, from \eqref{eq:A-2d} we identify
\begin{align}
a(x,y)
    &= \frac{1}{2} \del^2 \ee^{2(\beta-1)x}, &
b(x,y)
    &=  0, &
c(x,y)
    &=  0, &
f(x,y)
    &=  0 .
\end{align}
We fix a time to maturity $t$ and $\log$-strike $k$.  Using the formulas from Appendix \ref{sec:impvol.2} as well as the Mathematica notebook provided on the authors' website,  we compute explicitly
\begin{align}
\begin{aligned}
\sig_{0}
    &= \del \, \ee^{(\beta-1)x}, \\
\sig_{1}
    &= \frac{1}{2} (\beta-1 ) \sig_{0} (k-x) , \\
\sig_{2}
        &=  \frac{t }{24} (\beta -1)^2 \sigma _0^3
                - \frac{t^2 }{96} (\beta -1)^2 \sigma _0^5
                + \frac{1}{12} (\beta -1)^2 \sigma _0(k-x)^2, \\
\sig_{3}
        &=  \frac{t}{16} (\beta -1)^3 \sigma _0^3 (k-x)
                + \frac{-5t^2}{192} (\beta -1)^3 \sigma _0^5 (k-x)
\end{aligned} \label{eq:sig.CEV}
\end{align}
In the CEV setting the exact price of a Call option is derived in \citet{CoxCEV}:
\begin{align}
\begin{aligned}
u(t,x)
    &=  e^x Q(\kappa, 2+\tfrac{2}{2-\beta}, 2\chi) - \ee^k \( 1 - Q(2\chi,\tfrac{2}{2-\beta},2\kappa) \) ,  \\
Q(w,v,\mu)
    &=  \sum_{n=0}^{\infty} \( \frac{(\mu /2)^n \ee^{-\mu /2}}{n!} \frac{\Gam(v/2+n,w/2)}{\Gam(v/2+n)} \) , \\
\chi
    &=  \frac{2 \ee^{(2-\beta)x}}{\del^2 (2-\beta)^2 t} , \\
\kappa
    &=  \frac{2 \ee^{(2-\beta)k}}{\del^2 (2-\beta)^2 t} ,
\end{aligned} \label{eq:u.CEV}
\end{align}
where $\Gam(a)$ and $\Gam(a,b)$ denote the complete and incomplete Gamma functions respectively.
Thus, the implied volatility $\sig$ can be obtained numerically by solving \eqref{eq:imp.vol.def}.  In Figure \ref{fig:cev} we plot our {third} order implied volatility approximation
{$\bar{\sig}_3$} and the numerically obtained implied volatility $\sig$.  For comparison, we also plot the implied
volatility expansion of \cite{hagan-woodward}
\begin{align}
\sig^{\text{HW}}
    &=  \frac{\del}{f^{1-\beta}}\( 1 + \frac{(1-\beta)(2+\beta)}{24}\(\frac{e^x-e^k}{f}\)^2
            + \frac{(1-\beta)^2}{24} \frac{\del^2 t}{f^{2(1-\beta)}} + \cdots \) , &
f
    &=  \frac{1}{2}(\ee^x+\ee^k) . \label{eq:sig.hw}
\end{align}

\subsection{Heston stochastic volatility model}
\label{sec:Heston}
Perhaps the most well-known stochastic volatility model is that of \cite{heston1993}.  In the Heston model, the dynamics of the underlying $S$ are given by
\begin{align}
\dd S_t
    &=  \sqrt{Z_t} S_t \dd W_t , &
S_0
    &=  s > 0 , \\
\dd Z_t
    &=  \kappa (\theta - Z_t) \dd t + \del \sqrt{Z_t} \dd B_t , &
Z_0
    &=  z > 0 , \\
\dd\<W,B\>_t
    &=  \rho \, \dd t .
\end{align}
As pointed out in \cite{andersen2007moment}, one must set $\rho<0$ in order to prevent a moment
explosion. In order to improve the efficacy of our approximation it is convenient to perform the
following change of variable $(X_t,V_t):=(\log S, e^{\kappa t }Z_t)$.  Changing from $Z$ to $V$
removes the geometric part of the drift (see also \cite{BompisGobet2012}). 
By Ito's formula we obtain
\begin{align}
\dd X_t
    &=  -\frac{1}{2} e^{-\kappa  t}V_t \dd t +\sqrt{e^{-\kappa  t} V_t} \dd W_t , &
X_0
    &=  x := \log s , \\\dd V_t
    &=\theta  \kappa\,  e^{\kappa  t}  \dd t + \delta \sqrt{e^{\kappa  t}V_t} \dd B_t , &
V_0
    &=  v:= z > 0 , \\
\dd\<W,B\>_t
    &=  \rho \, \dd t . \label{eq:model.Heston}
\end{align}
The generator of $(X,V)$ is given by
\begin{align}
\Ac(t)
    &=  \frac{1}{2} e^{-\kappa  t}v \( \d_x^2 - \d_x \) + \theta  \kappa\,  e^{\kappa  t}\, \d_v
             + \frac{1}{2} \del^2  \delta e^{\kappa  t}v  \, \d_v^2 + \delta  \rho  v\, \d_x \d_v .
\end{align}
Thus, using \eqref{eq:A-2d}, we identify
\begin{align}
a(x,v)
    &= \frac{1}{2} e^{-\kappa  t}v, &
b(x,v)
    &=  \del^2  \delta e^{\kappa  t}v, &
c(x,v)
    &= \delta  \rho  v , &
f(x,v)
    &=  \theta  \kappa\,  e^{\kappa  t} .
\end{align}
We fix a time to maturity $t$, a $\log$-strike $k$, and we consider the time-dependent Taylor series expansion of $\Ac(t)$ as described in Example \ref{example:TimeTaylor} with $(\xb(t),\bar{v}(t))=(X_0,\mathbb{E}[V_t]):=(x,\theta  \left(e^{\kappa  t}-1\right))$.  Using the Mathematica notebook provided on the authors' website,  we
compute explicitly
\begin{align}
\sig_0
    &= \sqrt{\frac{-\theta +\theta  \kappa  t+e^{-\kappa  t} (\theta -v)+v}{\kappa  t}}, \\
\sig_{1}
    &=  \frac{\delta  \rho  z e^{-\kappa  t} \left(-2 \theta -\theta  \kappa  t-e^{\kappa  t} (\theta  (\kappa  t-2)+v)+\kappa  t v+v\right)}{\sqrt{2} \kappa ^2 \sigma_0^2 t^{3/2}}, \\
\sig_{2}
    &= \frac{\delta ^2 e^{-2 \kappa  t}}{32 \kappa ^4 \sigma_0^5 t^3}\Bigg(
    -2 \sqrt{2} \kappa  \sigma_0^3 t^{3/2} z \left(-\theta -4 e^{\kappa  t} (\theta +\kappa  t (\theta -v))+e^{2 \kappa  t} (\theta  (5-2 \kappa  t)-2 v)+2 v\right)
        \label{eq:sig.Heston} \\
 &\  +\kappa  \sigma_0^2 t \left(4 z^2-2\right) \left(\theta +e^{2 \kappa  t} \left(-5 \theta +2 \theta  \kappa  t+8 \rho ^2 (\theta  (\kappa  t-3)+v)+2 v\right)\right)\\
&\ +\kappa  \sigma_0^2 t \left(4 z^2-2\right)\left(4 e^{\kappa  t} \left(\theta +\theta  \kappa  t+\rho ^2 (\theta  (\kappa  t (\kappa  t+4)+6)-v (\kappa  t (\kappa  t+2)+2))-\kappa  t v\right)-2 v\right)\\
&\ +4 \sqrt{2} \rho ^2 \sigma_0 \sqrt{t} z \left(2 z^2-3\right) \left(-2 \theta -\theta  \kappa  t-e^{\kappa  t} (\theta  (\kappa  t-2)+v)+\kappa  t v+v\right)^2\\
&\ +4 \rho ^2 \left(4 \left(z^2-3\right) z^2+3\right) \left(-2 \theta -\theta  \kappa  t-e^{\kappa  t} (\theta  (\kappa  t-2)+v)+\kappa  t v+v\right)^2
    \Bigg)-\frac{\sigma_1^2 \left(4 (x-k)^2-\sigma_0^4 t^2\right)}{8 \sigma_0^3 t},\\
z&=\frac{x-k-\frac{\sigma_0^2 t}{2}}{\sigma_0 \sqrt{2 t}}.
\end{align}
The expression for $\sig_3$ is too long to reasonably put in the text.  However, the explicit form of $\sig_3$ is provided in the Mathematica notebook on the authors' website.
\par
The characteristic function of $X_t$ is computed explicitly in \cite{heston1993}
\begin{align}
\eta(t,x,y,\lam)
    := \log \Eb_{x,y} \ee^{ \ii \lam X_t}
    &=  {\ii \lam x + C(t,\lam) + D(t,\lam) \ee^y} , \\
C(t,\lam)
    &=  \frac{\kappa \theta}{ \del^2} \( (\kappa - \rho \delta  \ii \lam + d(\lam) ) t
            -2 \log \[ \frac{1-f(\lam) \ee^{d(\lam)t}}{1-f(\lam)}\]\) , \\
D(t,\lam)
    &=  \frac{\kappa - \rho \del \ii \lam + d(\lam)}{\del^2} \frac{1-\ee^{d(\lam)t}}{1-f(\lam) \ee^{d(\lam)t}} , \\
f(\lam)
    &=  \frac{\kappa - \rho \del \ii \lam + d(\lam)}{\kappa - \rho \del \ii \lam - d(\lam)} , \\
d(\lam)
    &=  \sqrt{ \del^2 (\lam^2 + \ii \lam) + (\kappa - \rho \ii \lam \del)^2} .
\end{align}
Thus, the price of a European Call option can be computed using standard Fourier methods
\begin{align}
u(t,x,y)
    &=  \frac{1}{2\pi} \int_\Rb \dd \lam_r \, \ee^{\eta(t,x,y,\lam)} \widehat{\varphi}(\lam) , &
\widehat{\varphi}(\lam)
    &=  \frac{-\ee^{k-\ii k \lam}}{ \ii \lam + \lam^2 } , &
\lam
    &=  \lam_r + \ii \lam_i , &
\lam_i
    &<  -1 . \label{eq:u.Heston}
\end{align}
Note, since the Call option payoff $\varphi(x)=(\ee^x -\ee^k)^+$ is not in $L^1(\Rb)$, its Fourier
transform $\widehat{\varphi}(\lam)$ must be computed in a generalized sense by fixing an imaginary
component of the Fourier variable $\lam_i < -1$. Using \eqref{eq:u.Heston} the implied volatility
$\sig$ can be computed to solving \eqref{eq:imp.vol.def} numerically.  In Figure \ref{fig:heston}
we plot our third order implied volatility approximation {$\bar{\sig}_3$} and the numerically
obtained implied volatility $\sig$.  For comparison, we also plot the small-time near-the-money
implied volatility expansion of \cite{forde-jacquier-lee} (see Theorem 3.2 and Corollary 4.3)
\begin{align}
\sig^{\text{FJL}}
    &=  \( g_0^2 + g_1 \, t + o(t) \)^{1/2} , \label{eq:sig.fjl} \\
g_0
    &=  \ee^{y/2} \( 1 + \frac{1}{4} \rho \del (k-x) \ee^{-y} + \frac{1}{24}\( 1- \frac{5 \rho^2}{2} \)
            \del^2 (k-x)^2 \ee^{-2 y} \) + \Oc( (k-x)^3 ), \\
g_1
    &=  - \frac{ \del^2}{12} \( 1 - \frac{\rho^2}{4} \) + \frac{ \ee^y \rho \del}{4} + \frac{\kappa}{2}(\theta-\ee^y)
            + \frac{1}{24} \rho \del \ee^{-y} (\del^2 \rhob^2 - 2 \kappa(\theta + \ee^y) + \rho \del \ee^y )(k-x) \\ &\qquad
            + \frac{\del^2 \ee^{-2y} }{7680} \( 176 \del^2 - 480 \kappa \theta - 712 \rho^2 \del^2 + 521 \rho^4 \del^2
            + 40 \rho^3 \del \ee^y + 1040 \kappa \theta \rho^2 - 80 \kappa \rho^2 \ee^y \)(k-x)^2 \\ & \qquad
            + \Oc((k-x)^3 ), \qquad \qquad
\rhob
    =   \sqrt{1-\rho^2} .
\end{align}


\subsection{$3/2$ stochastic volatility model}
\label{sec:Three-Halves}
We consider now the 3/2 stochastic volatility model.  The risk-neutral dynamics of the underlying $S$ in this setting are given by
\begin{align}
\dd S_t
    &=  \sqrt{Z_t} S_t \dd W_t , &
S_0
    &=  s > 0 , \\
\dd Z_t
    &=  Z_t \( \kappa (\theta - Z_t) \dd t + \del \sqrt{Z_t} \dd B_t \) , &
Z_0
    &=  z > 0 , \\
\dd\<W,B\>_t
    &=  \rho \, \dd t .
\end{align}
As in all stochastic volatility models, one typically sets $\rho<0$ in order to capture the
leverage effect. The 3/2 model is noteworthy in that it does not fall into the affine class of
\cite{duffiepansingleton}, and yet it still allows for European option prices to be computed in
semi-closed form (as a Fourier integral).  Notice however that the characteristic function (given in \eqref{eq:char.three-half} below) involves
special functions such as the Gamma and the confluent hypergeometric functions.  Therefore,
Fourier pricing methods are not an efficient means of computed prices.
The importance of the 3/2 model in the pricing of options on realized variance is well documented by \cite{Drimus}. In particular, the 3/2 model
allows for upward-sloping implied volatility of variance smiles while Heston's model leads to
downward-sloping volatility of variance smiles, in disagreement with observed skews in variance markets.
\par
In $\log$ notation $(X,Y) := (\log S , \log Z)$ we have the following dynamics
\begin{align}
\begin{aligned}
\dd X_t
    &=  -\frac{1}{2} \ee^{Y_t} \dd t + \ee^{\tfrac{1}{2}Y_t} \dd W_t , &
X_0
    &=  x := \log s , \\
\dd Y_t
    &=  \( \kappa (\theta - \ee^{Y_t} ) - \frac{1}{2}\del^2 \ee^{Y_t} \) \dd t + \del \, \ee^{\tfrac{1}{2}Y_t} \dd B_t , &
Y_0
    &=  y := \log z , \\
\dd\<W,B\>_t
    &=  \rho \, \dd t .
\end{aligned} \label{eq:model.three-halves}
\end{align}
The generator of $(X,Y)$ is given by
\begin{align}
\Ac
    &=  \frac{1}{2} \ee^{y} \( \d_x^2 - \d_x \) + \( \kappa (\theta - \ee^y) - \frac{1}{2}\del^2 \ee^y \) \d_y
             + \frac{1}{2} \del^2 \ee^{y} \d_y^2 + \rho \, \del \, \ee^y \d_x \d_y .
\end{align}
Thus, using \eqref{eq:A-2d}, we identify
\begin{align}
a(x,y)
    &= \frac{1}{2} \ee^{y}, &
b(x,y)
    &=  \frac{1}{2} \del^2 \ee^{y}, &
c(x,y)
    &=  \rho \, \del \, \ee^{y} , &
f(x,y)
    &=  \kappa (\theta - \ee^y) - \frac{1}{2}\del^2 \ee^y .
\end{align}
We fix a time to maturity $t$ and $\log$-strike $k$.  Using the formulas from Appendix
\ref{sec:impvol.2} as well as the Mathematica notebook provided on the authors' website, we
compute explicitly
\begin{align}
\sig_0
  &= \ee^{y/2}, \\
\sig_1
    &= \frac{t}{8} \left(2 \theta  \kappa  \sigma _0-\sigma _0^3 \left(\delta ^2-\delta  \rho +2 \kappa \right)\right)
            + \frac{1}{4} \delta  \rho  \sigma _0 (k-x) , \\
\sig_2
    &=  \frac{t}{96} \delta ^2 \left(8-7 \rho ^2\right) \sigma _0^3 \\ &\qquad
            + \frac{t^2}{384} \left(-36 \theta  \kappa  \sigma _0^3 \left(\delta ^2-\delta  \rho +2 \kappa \right)+\sigma _0^5 \left(13 \delta ^4-26 \delta ^3 \rho +4 \delta ^2 \left(13 \kappa +4 \rho ^2-1\right)-52 \delta  \kappa  \rho +52 \kappa ^2\right)+20 \theta ^2 \kappa ^2 \sigma _0\right) \\[-1.5em] &\qquad
            + \frac{t}{96} \delta  \rho  \sigma _0 \left(6 \theta  \kappa -7 \sigma _0^2 \left(\delta ^2-\delta  \rho +2 \kappa \right)\right) (k-x)
            - \frac{1}{48} \delta ^2 \left(\rho ^2-2\right) \sigma _0 (k-x)^2 , \\
\sig_3
    &=  \frac{t^2}{256} \delta ^2 \sigma _0^3 \left(5 \left(3 \rho ^2-4\right) \sigma _0^2 \left(\delta ^2-\delta  \rho +2 \kappa \right)+2 \theta  \kappa  \left(8-7 \rho ^2\right)\right) \label{eq:sig.three-halves} \\ &\qquad
            + \frac{t^3}{3072} \Big(
            -132 \theta ^2 \kappa ^2 \sigma _0^3 \left(\delta ^2-\delta  \rho +2 \kappa \right)+10 \theta  \kappa  \sigma _0^5 \left(13 \delta ^4-26 \delta ^3 \rho +4 \delta ^2 \left(13 \kappa +4 \rho ^2-1\right)-52 \delta  \kappa  \rho +52 \kappa ^2\right)
            \\[-1.5em] & \qquad \qquad
            + 24 \theta ^3 \kappa ^3 \sigma _0
            -\sigma _0^7 \left(\delta ^2-\delta  \rho +2 \kappa \right) \left(35 \delta ^4-70 \delta ^3 \rho +2 \delta ^2 \left(70 \kappa +29 \rho ^2-16\right)-140 \delta  \kappa  \rho +140 \kappa ^2\right)
            \Big)\\ &\qquad
            + \frac{t}{128} \delta ^3 \rho  \left(4-3 \rho ^2\right) \sigma _0^3 (k-x)
            + \frac{t^2 \delta  \rho  \sigma _0}{1536} \Big( -84 \theta  \kappa  \sigma _0^2 \left(\delta ^2-\delta  \rho +2 \kappa \right) \Big) (k-x) \\ &\qquad
            + \frac{t^2 \delta  \rho  \sigma _0}{1536}\Big(  +\sigma _0^4 \left(45 \delta ^4-90 \delta ^3 \rho +4 \delta ^2 \left(45 \kappa +14 \rho ^2-4\right)-180 \delta  \kappa  \rho +180 \kappa ^2\right)+20 \theta ^2 \kappa ^2 \Big) (k-x) \\ &\qquad
            + \frac{t}{384} \delta ^2 \sigma _0 \left(\left(\rho ^2-8\right) \sigma _0^2 \left(\delta ^2-\delta  \rho +2 \kappa \right)-2 \theta  \kappa  \left(\rho ^2-2\right)\right) (k-x)^2 ,
\end{align}
To the best of our knowledge, the above formula is the first
explicit implied volatility expansion for the 3/2 model.
The characteristic function of $X_t$ is given, for example, in Proposition 3.2 of \cite{baldeaux2012consistent}.  We have
\begin{align}
\Eb_{x,y} \ee^{\ii \lam X_t}
    &=  \ee^{\ii \lam x} \frac{\Gamma(\gamma - f)}{\Gamma(\gamma)} \( \frac{2}{\del^2 z} \)^f
            M\( f, \gamma, \frac{-2}{\del^2 z}\) , &
z
    &=  \frac{\ee^y}{\kappa \theta}(\ee^{\kappa  \theta  t}-1), &
\gamma
    &=  2 \( f + 1 - \frac{p}{\del^2} \), \label{eq:char.three-half} \\
f
    &=  -\( \frac{1}{2} - \frac{p}{\del^2} \) + \( \( \frac{1}{2} - \frac{p}{\del^2} \)^2 + 2 \frac{q}{\del^2} \)^{1/2}, &
p
    &= - \kappa + \ii \del \rho \lam , &
q
    &= \frac{1}{2} (\ii \lam + \lam^2) ,
\end{align}
where $\Gamma$ is a Gamma function and $M$ is a confluent hypergeometric function. Thus, the price
of a European Call option can be computed using standard Fourier methods
\begin{align}
u(t,x,y)
    &=  \frac{1}{2\pi} \int_\Rb \dd \lam_r \, \widehat{\varphi}(\lam) \Eb_{x,y} \ee^{\ii \lam X_t} , &
\lam
    &=  \lam_r + \ii \lam_i , &
\lam_i
    &<  -1 , \label{eq:u.three-halves}
\end{align}
where $\widehat{\varphi}(\lam)$ is given in \eqref{eq:u.Heston}.  Using \eqref{eq:u.three-halves} the implied volatility $\sig$ can be computed to solving \eqref{eq:imp.vol.def} numerically.
In Figure \ref{fig:three-halves} we plot our {third} order implied volatility approximation
{$\bar{\sig}_3$} and the numerically obtained implied volatility $\sig$.


\subsection{SABR local-stochastic volatility}
\label{sec:SABR}
The SABR model of \cite{sabr} is a local-stochastic volatility model in which the risk-neutral dynamics of $S$ are given by
\begin{align}
\dd S_t
    &=  Z_t S_t^\beta \dd W_t , &
S_0
    &=  s > 0, \\
\dd Z_t
    &=  \del Z_t \dd B_t , &
Z_0
    &=  z > 0, \\
\dd\<W,B\>_t
    &=  \rho \, \dd t .
\end{align}
Modeling $Z$ as a geometric Brownian motion results in a
true implied volatility smile (i.e., upward sloping implied volatility for high strikes); this is
in contrast to the CEV model, for which the model-induced implied volatility is monotone decreasing
(for $\beta <1$).  In $\log$ notation $(X,Y) := (\log S , \log Z)$ we have, we have the following
dynamics:
\begin{align}
\begin{aligned}
\dd X_t
    &=  -\frac{1}{2} \ee^{2Y_t + 2(\beta-1)X_t} \dd t + \ee^{Y_t + (\beta-1)X_t} \dd W_t , &
X_0
    &=  x := \log s , \\
\dd Y_t
    &=  -\frac{1}{2} \del^2 \dd t + \del \, \dd B_t , &
Y_0
    &=  y := \log z , \\
\dd\<W,B\>_t
    &=  \rho \, \dd t .
\end{aligned} \label{eq:model.SABR}
\end{align}
The generator of $(X,Y)$ is given by
\begin{align}
\Ac
    &=  \frac{1}{2} \ee^{2y + 2(\beta-1)x} (\d_x^2 - \d_x) - \frac{1}{2} \del^2 \d_y
            + \frac{1}{2} \del^2 \d_y^2 + \rho \, \del \, \ee^{y + (\beta-1)x} \d_x \d_y .
\end{align}
Thus, using \eqref{eq:A-2d}, we identify
\begin{align}
a(x,y)
    &= \frac{1}{2} \ee^{2y + 2(\beta-1)x}, &
b(x,y)
    &=  \frac{1}{2} \del^2 , &
c(x,y)
    &=  \rho \, \del \, \ee^{y + (\beta-1)x}, &
f(x,y)
    &=  - \frac{1}{2} \del^2 .
\end{align}
We fix a time to maturity $t$ and $\log$-strike $k$.  Using the formulas from Appendix
\ref{sec:impvol.2} as well as the Mathematica notebook provided on the authors' website,  we
compute explicitly
\begin{align}
\sig_0
    &= \ee^{y + (\beta-1) x }, &
\sig_1
    &=  \sig_{1,0} + \sig_{0,1} , &
\sig_2
    &=  \sig_{2,0} + \sig_{1,1} + \sig_{0,2} , &
\sig_3
    &=  \sig_{3,0} + \sig_{2,1} + \sig_{1,2} + \sig_{0,3} , \label{eq:sig.SABR}
\end{align}
where
\begin{align}
\sig_{1,0}
    &=  \frac{1}{2} (k-x) (-1+\beta ) \sig_0 , \\
\sig_{0,1}
    &= \frac{1}{4} \del  \(2 (k-x) \rho +t \sig_0 \(-\delta +\rho  \sigma_0\)\), \\
\sig_{2,0}
        &=  \frac{t}{24} (\beta -1)^2 \sigma _0^3 -\frac{t^2}{96} (\beta -1)^2 \sigma _0^5 + \frac{1}{12} (\beta -1)^2 \sigma _0 (k-x)^2, \\
\sig_{1,1}
        &=  \frac{t}{12} (\beta -1) \delta  \rho  \sigma _0^2 -\frac{t^2}{48} (\beta -1) \delta  \rho  \sigma _0^4
                + \frac{t}{24} (\beta -1) \delta  \sigma _0 \left(\delta +\rho  \sigma _0\right) (k-x)
                -\frac{1}{3} (\beta -1) \delta  \rho (k-x)^2, \\
\sig_{0,2}
        &=  \frac{t}{24} \delta ^2 \left(8-3 \rho ^2\right) \sigma _0
                + \frac{t^2}{96} \delta ^2 \sigma _0 \left(5 \delta ^2+2 \sigma _0 \left(\left(6 \rho ^2-2\right) \sigma _0-7 \delta  \rho \right)\right) \\ & \qquad
                - \frac{t}{24} \delta ^2 \rho  \left(\delta -3 \rho  \sigma _0\right) (k-x)
                + \frac{\delta ^2 \left(2-3 \rho ^2\right)}{12 \sigma _0} (k-x)^2 , \\
\sig_{3,0}
        &=  \frac{t}{16} (\beta -1)^3 \sigma _0^3 (k-x)
                - \frac{-5t^2}{192} (\beta -1)^3 \sigma _0^5 (k-x) , \\
\sig_{2,1}
        &=  \frac{t^2}{288} (\beta -1)^2 \delta  \sigma _0^3 \left(17 \rho  \sigma _0-11 \delta \right)
                + \frac{t^3}{384} (\beta -1)^2 \delta  \sigma _0^5 \left(3 \delta -5 \rho  \sigma _0\right)
                + \frac{t}{16} (\beta -1)^2 \delta  \rho  \sigma _0^2 (k-x) \\ &\qquad
                + \frac{-3t^2}{64}(\beta -1)^2 \delta  \rho  \sigma _0^4 (k-x)
                + \frac{t}{48} (\beta -1)^2 \delta  \sigma _0 \left(\rho  \sigma _0-2 \delta \right) (k-x)^2
                + \frac{5}{24} (\beta -1)^2 \delta  \rho (k-x)^3, \\
\sig_{1,2}
        &=  -\frac{t^2}{72} (\beta -1) \delta ^2 \rho  \sigma _0^2 \left(\delta -7 \rho  \sigma _0\right)
                + \frac{t^3}{96} (\beta -1) \delta ^2 \rho  \sigma _0^4 \left(2 \delta -3 \rho  \sigma _0\right) \\ &\qquad
                + \frac{t}{144} (\beta -1) \delta ^2 \left(2-17 \rho ^2\right) \sigma _0 (k-x)
                + \frac{t^2}{192} (\beta -1) \delta ^2 \sigma _0 \left(\delta ^2-6 \delta  \rho  \sigma _0+2 \left(\rho ^2-1\right) \sigma _0^2\right)(k-x) \\ &\qquad
                + \frac{t}{48} (\beta -1) \delta ^2 \rho  \left(5 \rho  \sigma _0-7 \delta \right)(k-x)^2
                + \frac{(\beta -1) \delta ^2 \left(16 \rho ^2-7\right)}{24 \sigma _0} (k-x)^3 , \\
\sig_{0,3}
        &=  \frac{t^2}{96} \delta ^3 \sigma _0 \left(3 \delta  \left(\rho ^2-4\right)+\rho  \left(26-9 \rho ^2\right) \sigma _0\right) \\ &\qquad
                + \frac{t^3}{384} \delta ^3 \sigma _0 \left(\sigma _0 \left(19 \delta ^2 \rho +2 \sigma _0 \left(\delta  \left(8-21 \rho ^2\right)+\rho  \left(15 \rho ^2-11\right) \sigma _0\right)\right)-3 \delta ^3\right) \\ &\qquad
                +   \frac{t}{48} \delta ^3 \rho  \left(3 \rho ^2-2\right) (k-x)
                -   \frac{t^2}{192} \delta ^3 \rho  \left(\delta ^2+6 \sigma _0 \left(\delta  \rho +\left(1-2 \rho ^2\right) \sigma _0\right)\right)(k-x)\\ &\qquad
                - \frac{t}{16} \delta ^3 \rho  \left(\rho ^2-1\right) (k-x)^2
                + \frac{\delta ^3 \rho  \left(6 \rho ^2-5\right)}{24 \sigma _0^2} (k-x)^3 ,
\end{align}
There is no formula for European option prices in the general SABR setting.  However, for the
special zero-correlation case $\rho = 0$ the exact price of a European Call is computed in
\cite{sabr-exact}:
\begin{align}
\begin{aligned}
u(t,x)
    &=  \ee^{(x+k)/2} \frac{\ee^{-\del^2 t / 8}}{\sqrt{2 \pi \del^2 t}} \Bigg\{
    \frac{1}{\pi} \int_0^\infty \dd V \int_0^\pi \dd\phi \frac{1}{V} \( \frac{V}{V_0} \)^{-1/2}
    \frac{\sin \phi \sin (|\nu|\phi)}{b-\cos \phi} \exp \( \frac{\xi_\phi^2}{2 \del^2 t}\) \\ & \qquad
    + \frac{\sin (|\nu|\pi)}{\pi} \int_0^\infty \dd V \int_0^\infty \dd\psi \frac{1}{V} \( \frac{V}{V_0} \)^{-1/2}
    \frac{\sinh \psi}{b - \cosh \psi}\ee^{-|\nu|\psi}\exp \( \frac{\xi_\psi^2}{2 \del^2 t } \)
    \Bigg\} + ( \ee^x - \ee^k )^+, \\
\xi_\phi
    &= \arccos \( \frac{q_h^2 + q_x^2 + V^2 + V_0^2}{2 V V_0} - \frac{q_h q_x}{V V_0} \cos\phi \), \\
\xi_\psi
    &= \arccos \( \frac{q_h^2 + q_x^2 + V^2 + V_0^2}{2 V V_0} + \frac{q_h q_x}{V V_0} \cosh \psi \), \\
b
    &=  \frac{q_h^2 + q_x^2}{2 q_h q_x} , \qquad
q_h
    =   \frac{\ee^{(1-\beta)k}}{1-\beta} , \qquad
q_x
    =       \frac{\ee^{(1-\beta)x}}{1-\beta}    , \qquad
\nu
    =       \frac{-1}{2(1-\beta)} , \qquad
V_0
    =       \frac{\ee^y}{\del} .
\end{aligned} \label{eq:u.SABR}
\end{align}
Thus, in the zero-correlation setting, the implied volatility $\sig$ can be obtained by
using the above formula and then by solving \eqref{eq:imp.vol.def} numerically.  In Figure \ref{fig:sabr}
we plot our third order implied volatility approximation $\bar{\sig}_3$ and the numerically obtained implied
volatility $\sig$.  For comparison, we also plot the implied volatility expansion of \cite{sabr}
\begin{align}
\sig^{\text{HKLW}}
    &=  \del \frac{x-k}{D(\zeta)} \left\{ 1 + t \del^2
            \[ \frac{2 \gam_2 - \gam_1^2 + 1/f^2}{24} \(\frac{\ee^{y+\beta f}}{\del}\)^2
            + \frac{\rho \gam_1 \ee^{y+\beta f}}{4 \del} + \frac{2-3\rho^2}{24} \] \right\} , \label{eq:sig.hklw}  \\
f
    &=  \frac{1}{2}(\ee^x+\ee^k), \\
\zeta
    &= \frac{\del \, \ee^{-y}}{\beta-1}\( \ee^{(1-\beta)k} - \ee^{(1-\beta)x} \) , \\
\gam_1
    &= \beta / f, \\
\gam_2
    &= \beta(\beta-1)/f^2, \\
D(\zeta)
    &=  \log \( \frac{\sqrt{1-2 \rho \zeta + \zeta^2}+\zeta-\rho}{1-\rho} \) .
\end{align}
Note that we use the ``corrected'' SABR formula, which appears in \cite{obloj}.

%
%

\section{Conclusions and future work}
\label{sec:conclusion}
In this paper we consider a general class of parametric local-stochastic volatility models.   In
this setting, we provide a family of approximations -- one for each polynomial expansion of $\Ac(t)$ -- for (i)  European-style option prices and (ii) implied volatilities.  The terms in our option price expansions are expressed as a differential operator acting on the Black-Scholes price.  Thus, to compute approximate prices, one
requires only a normal CDF.  Our implied volatility expansions are explicit, requiring no special
functions nor any numerical integration.  Thus, approximate implied volatilities can be computed even faster than option prices.
\par
We carry out extensive computations using the Taylor series expansion of $\Ac(t)$.  In particular,
we establish the rigorous error bounds of our pricing {and implied volatility approximations}.  We
also implement our implied volatility expansion under four separate model dynamics: CEV local
volatility, Heston stochastic volatility, 3/2 stochastic volatility, and SABR local-stochastic
volatility.  In each setting we demonstrate that our implied volatility expansion provides an
excellent approximation of the true implied volatility over a large range of strikes and
maturities.

\subsection*{Thanks}
The authors would like to thank Mike Staunton and two anonymous referees for their thorough reading of this manuscript.  Their suggestions have improved both the mathematical quality and readability of our results.

%
%

\appendix

%
%

\section{Asymptotics of the Black-Scholes price for short maturities}
\label{sec:lemmas} We prove some results concerning the short-maturity behavior of the
Black-Scholes price. Throughout this appendix $\tau$ denotes the time to maturity.  We recall the
following {alternative expression for the Black-Scholes price,} taken from
\cite{roper2009relationship}
\begin{align}\label{andre1}
  u^{\BS}(\s;\t,x,k)&=\left(e^{x}-e^{k}\right)^{+}+
  e^{x}\sqrt{\frac{\t}{2\pi}}\int_{0}^{\s}e^{-\frac{1}{2}\left(\frac{x-k}{w\sqrt{\t}}+\frac{w\sqrt{\t}}{2}\right)^{2}}dw.
\end{align}
Now we set 
\begin{align}\label{andre3}
F\left(\s_{1},\s_{2},\t,\lam\right)
    &:=\int_{\s_{1}}^{\s_{2}}e^{-\frac{1}{2}\left(\frac{\lam^{2}}{w^{2}}+\t\frac{w^{2}}{4}\right)}dw, &
\s_{1}
    &\le \s_{2}.
\end{align}
and observe that, if
\begin{align}\label{andre5}
 |x-k|\le \lam \sqrt{\t}
\end{align}
for some $\lam>0$, then we have
\begin{align}
  e^{-\frac{\lam \sqrt{\t}}{2}}F\left(\s_1,\s_2,\t,\lam\right)&=\int_{\s_1}^{\s_2}e^{-\frac{1}{2}\left(\frac{\lam}{w}+\frac{w\sqrt{\t}}{2}\right)^{2}}dw\le
  \int_{\s_1}^{\s_2}e^{-\frac{1}{2}\left(\frac{x-k}{w\sqrt{\t}}+\frac{w\sqrt{\t}}{2}\right)^{2}}dw\\
  &\le
 \int_{\s_1}^{\s_2}e^{-\frac{1}{2}\left(\frac{\lam}{w}-\frac{w\sqrt{\t}}{2}\right)^{2}}dw=
 e^{\frac{\lam \sqrt{\t}}{2}}F\left(\s_1,\s_2,\t,\lam\right). \label{eq:matt3}
\end{align}
Therefore, assuming \eqref{andre5} holds, from \eqref{andre1} and \eqref{eq:matt3} we have
\begin{align}\label{andre2}
  e^{x-\frac{\lam \sqrt{\t}}{2}}\sqrt{\frac{\t}{2\pi}}F\left(0,\s,\t,\lam\right)\le u^{\BS}(\s;\t,x,k) - \left(e^{x}-e^{k}\right)^{+}\le
  e^{x+\frac{\lam \sqrt{\t}}{2}}\sqrt{\frac{\t}{2\pi}}F\left(0,\s,\t,\lam\right).
\end{align}
Note that $F$ in \eqref{andre3} is a monotone function, increasing in $\s_{2}$, decreasing in
$\s_{1}$, $\t$ and $\lam$. In particular, for any $0\le \s_{\text{min}}\le\s_{\text{max}}$,
$\lam>0$, ${\t_0}>0$ and $\tau \in [0,{\t_0}]$, we have
\begin{align}\label{andre4}
  0<F\left(\s_{\text{min}},\s_{2},T,\lam\right)\le F\left(\s_{1},\s_{2},\t,\lam\right)\le F\left(\s_{1},\s_{\text{max}},0,\lam\right)<\infty ,\qquad
  \s_{\text{min}}\le\s_{1}\le\s_{2}\le \s_{\text{max}}.
\end{align}
The estimates in \eqref{andre2} were used by \cite{roper2009relationship} {(see also
\cite{Li2005})} to derive the asymptotic behavior close to expiry of the Black-Scholes Call price
as $\t\downarrow 0$. Below we use \eqref{andre2} to prove two lemmas concerning the comparison, close to
expiry, of two Black-Scholes prices with different volatilities.


\begin{lemma}
\label{landre1}
  For any $\lam >0$, $\s_{2}\ge \s_{1}>0$ and ${\t_0}>0$ there exists a
  constant $C\ge 1$, dependent only on $\lam,\s_{1},\s_{2}$ and ${\t_0}$, such that
\begin{align}\label{andre6}
  u^{\BS}(\s_{2};\t,x,k)\le C u^{\BS}(\s_{1};\t,x,k)
\end{align}
for any $\t\in[0,{\t_0}]$ and $|x-k|\le \lam \sqrt{\t}$.
\end{lemma}
\proof It suffices to prove that
\begin{align}\label{andre10}
  u^{\BS}(\s_{2};\t,x,k)- \left(e^{x}-e^{k}\right)^{+} \le C \left(u^{\BS}(\s_{1};\t,x,k)- \left(e^{x}-e^{k}\right)^{+}\right),
  \qquad |x-k|\le \lam \sqrt{\t},\quad \t\in[0,{\t_0}].
\end{align}
By \eqref{andre2} we have
\begin{align}
 u^{\BS}(\s_{2};\t,x,k)- \left(e^{x}-e^{k}\right)^{+} &\le   e^{x+\frac{\lam
 \sqrt{\t}}{2}}\sqrt{\frac{\t}{2\pi}}F\left(0,\s_{2},\t,\lam\right)\\
 &\le e^{\lam \sqrt{\t_{0}}}\frac{F\left(0,\s_{2},0,\lam\right)}{F\left(0,\s_{1},{\t_0},\lam\right)}\left(u^{\BS}(\s_{1};\t,x,k)-
 \left(e^{x}-e^{k}\right)^{+}\right),
\end{align}
where in the last inequality we used also \eqref{andre4}.
\endproof

\begin{lemma}\label{landre2}
For any $\lam >0$, $\s_{2}> \s_{1}>0$ and $C>0$ there exists ${\t_0}$ with $0<{\t_0}<\frac{1}{C}$,
dependent only on $\lam,\s_{1},\s_{2}$ and $C$, such that
\begin{align}\label{andre7}
 u^{\BS}(\s_{1};\t,x,k)
    &\le \left(1-C\t\right) u^{\BS}(\s_{2};\t,x,k),  \\
 { \left(1+C\t\right) u^{\BS}(\s_{1};\t,x,k) }  &\le { u^{\BS}(\s_{2};\t,x,k) , } \label{eq:second}
\end{align}
for any $\t\in[0,{\t_0}]$ and $|x-k|\le \lam \sqrt{\t}$.
\end{lemma}
\proof {To establish the first inequality \eqref{andre7}}, we prove that
\begin{align}\label{andre8}
  u^{\BS}(\s_{2};\t,x,k)-u^{\BS}(\s_{1};\t,x,k)\ge C\t u^{\BS}(\s_{2};\t,x,k),\qquad |x-k|\le \lam \sqrt{\t},\quad \t\in[0,{\t_0}].
\end{align}
We estimate the LHS in \eqref{andre8} using \eqref{andre1}. We have
\begin{align}
 u^{\BS}(\s_{2};\t,x,k)-u^{\BS}(\s_{1};\t,x,k)&= \sqrt{\t} \frac{e^{x}}{\sqrt{2\pi}}\int_{\s_{1}}^{\s_{2}}
 e^{-\frac{1}{2}\left(\frac{x-k}{w\sqrt{\t}}+\frac{w\sqrt{\t}}{2}\right)^{2}}dw\\
 &\ge \sqrt{\t}\, \frac{e^{x-\frac{\lam\sqrt{\t}}{2}}}{\sqrt{2\pi}}F\left(\s_{1},\s_{2},\t,\lam\right)\ge c e^{x}\sqrt{\t} , \label{andre9} \\
c
    &:= \frac{e^{-\frac{\lam}{2\sqrt{C}}}}{\sqrt{2\pi}}F\left(\s_{1},\s_{2},\frac{1}{C},\lam\right) , \label{eq:matt5}
\end{align}
where in the next-to-last inequality we used \eqref{andre5} and \eqref{eq:matt3},
and in the last inequality we used \eqref{andre4} and $\t<\frac{1}{C}$, so that
$c$
is positive and independent of $\t$.
Next, once again using
\eqref{andre2}, we can prove the following estimate for the RHS of \eqref{andre8}:
\begin{align}
u^{\BS}(\s_{2};\t,x,k)
 { \leq u^{\BS}\(\s_{2};\frac{1}{C},x,k\) }
\le e^{x}\left(1+e^{\frac{\lam }{2\sqrt{C}}}\sqrt{\frac{1}{2C\pi}}F\left(0,\s_{2},\frac{1}{C},\lam\right)\right) ,
\end{align}
and therefore, for $\t$ positive and suitably small, we have
  $$C\t u^{\BS}(\s_{2};\t,x,k)\le c e^{x}\sqrt{\t}$$
for $c$ as in \eqref{eq:matt5}. {This establishes the first inequality \eqref{andre7}.  To
establish the second inequality \eqref{eq:second} we have
\begin{align}
 (1+C\tau)u^{\BS}(\s_{1};\t,x,k)
    &\leq u^{\BS}(\s_{1};\t,x,k) + C \tau u^{\BS}(\s_{2};\t,x,k) \\
    &\leq   (1-C\tau)u^{\BS}(\s_{2};\t,x,k) + C \tau u^{\BS}(\s_{2};\t,x,k)
    =       u^{\BS}(\s_{2};\t,x,k) ,
\end{align}
where we have used \eqref{andre7} in the last inequality.
This concludes the proof.
}
\endproof

%
%

\section{Fa\`a di Bruno's formula and Bell polynomials}\label{append:faa_bell}
Here we briefly recall the well known Fa\`a di Bruno's formula (see \citet*{Riordan} and
\citet*{Johnson}), more precisely, its Bell polynomial version.  Let $f$ and $g$ be two $C^\infty$
real-valued functions on $\mathbb{R}$. The following representation holds:
\begin{align}\label{eq:Faa_di_Bruno_appendix}
\frac{\dd^n}{\dd x^n}f(g(x))=\sum_{h=1}^n f^{(h)}(g(x))\cdot \mathbf{B}_{n,h}\left(
\frac{\dd}{\dd x}g(x),\frac{\dd^2}{\dd x^2}g(x),\cdots,\frac{\dd^{n-h+1}}{\dd x^{n-h+1}}g(x)
\right),\qquad n\geq 1,
\end{align}
with $\mathbf{B}_{n,h}$ being the family of the Bell polynomials defined as
\begin{align}\label{eq:Bell_polyn_appendix}
\mathbf{B}_{n,h}(z)=\sum \frac{n!}{j_1! j_2! \cdots j_{n-h+1}!}\left( \frac{z_1}{1!} \right)^{j_1}
\left( \frac{z_2}{2!} \right)^{j_2}\cdots \left( \frac{z_{n-h+1}}{(n-h+1)!} \right)^{j_{n-h+1}},
\qquad 1\leq h \leq n,
\end{align}
where the sum is taken over all sequences $j_1, j_2,\cdots, j_{n-h+1}$ of non-negative integers
such that
\begin{align}\label{eq:relation_indexes_bell}
j_1+j_2+\cdots + j_{n-h+1}=h,\quad\text{and}\quad  j_1+2 j_2+\cdots +(n-h+1) j_{n-h+1}=n.
\end{align}

%
%

\section{Implied volatility expressions}
\label{sec:impvol.2} In this appendix we assume a time-homogeneous diffusion and use the Taylor
series expansion of $\Ac$ as in Example \ref{example:Taylor} with $(\xb,\yb)=(X_0,Y_0):=(x,y)$.
With $\Ac$ given by \eqref{eq:A-2d}, we introduce the notation
\begin{align}
\eta_{i,j}
    &=  \frac{\d_x^i\d_y^j \eta(\xb,\yb)}{i!j!} , &
\eta
    &\in    \{a, b, c, f \} .
\end{align}
and we compute, explicitly {(below $\t$ is time to maturity)}
\begin{align}
\sig_0
    &=  \sqrt{2 a_{0,0}} , &
\sig_1
    &=  \sig_{1,0} + \sig_{0,1} , &
\sig_2
    &=  \sig_{2,0} + \sig_{1,1} + \sig_{0,2} ,
\end{align}
where
\begin{align}
\sig_{1,0}
    &=  \(  \frac{a_{1,0}}{2 \sigma _0} \) (k-x), &
\sig_{0,1}
    &=  \tau \( \frac{a_{0,1} \left(c_{0,0}+2 f _{0,0}\right)}{4 \sigma _0} \)
            + \( \frac{a_{0,1} c_{0,0}}{2 \sigma _0^3} \) (k-x) ,
\end{align}
and
\begin{align}
\sig_{2,0}
    &=  \tau \Big( \frac{1}{12} \sigma _0 a_{2,0}-\frac{a_{1,0}^2}{8 \sigma _0} \Big)
                + \tau^2 \Big( -\frac{1}{96} \sigma _0 a_{1,0}^2 \Big)
                + \Big( \frac{2 \sigma _0^2 a_{2,0}-3 a_{1,0}^2}{12 \sigma _0^3} \Big) (k-x)^2 , \\
\sig_{1,1}
    &=  \frac{\tau}{12 \sigma _0^3} \Big( \sigma _0^2 a_{1,1} c_{0,0}+a_{0,1} \left(a_{1,0} c_{0,0}-2 \sigma _0^2 c_{1,0}\right) \Big)
                + \frac{\tau^2}{48 \sigma _0} \Big( -a_{0,1} a_{1,0} c_{0,0} \Big) \\ &\qquad
                + \frac{\tau}{24 \sigma _0^3} \Big( 2 \sigma _0^2 a_{1,1} \left(c_{0,0}+2 f_{0,0}\right)+a_{0,1} \left(2 \sigma _0^2 \left(c_{1,0}+2 f_{1,0}\right)-5 a_{1,0} \left(c_{0,0}+2 f_{0,0}\right)\right) \Big) (k-x)\\ &\qquad
                + \frac{1}{6 \sigma _0^5} \Big( \sigma _0^2 a_{1,1} c_{0,0}+a_{0,1} \left(\sigma _0^2 c_{1,0}-5 a_{1,0} c_{0,0}\right) \Big) (k-x)^2, \\
\sig_{0,2}
    &=   \frac{\tau}{24 \sigma _0^5} \Big( 4 \sigma _0^2 a_{0,2} \left(3 \sigma _0^2 b_{0,0}-c_{0,0}^2\right)+a_{0,1} \left(a_{0,1} \left(9 c_{0,0}^2-8 \sigma _0^2 b_{0,0}\right)-4 \sigma _0^2 c_{0,0} c_{0,1}\right) \Big) \\ &\qquad
                + \frac{\tau^2}{24 \sigma _0^3} \Big( a_{0,1} \left(-2 \sigma _0^2 a_{0,1} b_{0,0}+c_{0,0} \left(\sigma _0^2 \left(c_{0,1}+2 f_{0,1}\right)-3 a_{0,1} f_{0,0}\right)\right)
                    \\ & \qquad \qquad
                +   a_{0,1} f_{0,0} \left(2 \sigma _0^2 \left(c_{0,1}+2 f_{0,1}\right)-3 a_{0,1} f_{0,0}\right)
                +\sigma _0^2 a_{0,2} \left(c_{0,0}+2 f_{0,0}\right){}^2 \Big) \\ &\qquad
                + \frac{\tau}{24 \sigma _0^5} \Big( a_{0,1} \left(c_{0,0} \left(4 \sigma _0^2 \left(c_{0,1}+f_{0,1}\right)-18 a_{0,1} f_{0,0}\right)-9 a_{0,1} c_{0,0}^2+4 \sigma _0^2 c_{0,1} f_{0,0}\right)
                \\ & \qquad \qquad
                + 4 \sigma _0^2 a_{0,2} c_{0,0} \left(c_{0,0}+2 f_{0,0}\right)
                \Big) (k-x) \\ &\qquad
                + \frac{1}{12 \sigma _0^7} \Big( a_{0,1} \left(a_{0,1} \left(4 \sigma _0^2 b_{0,0}-9 c_{0,0}^2\right)+2 \sigma _0^2 c_{0,0} c_{0,1}\right)+2 \sigma _0^2 a_{0,2} c_{0,0}^2 \Big)(k-x)^2 ,
\end{align}
Higher order terms are too long to reasonably include in this text.  However, $\sig_3$ and (for
local volatility models) $\sig_4$ can be computed easily using the Mathematica code provided free
of charge on the authors' website.
\begin{verbatim}
http://explicitsolutions.wordpress.com
\end{verbatim}

\begin{footnotesize}
\bibliographystyle{chicago}
\bibliography{BibTeX-MasterA}
\end{footnotesize}

%
%


\begin{sidewaysfigure}
    \centering
    \begin{tabular}{ | c | c | }
        \hline
        $t=0.1$ & $t=1.0$ \\
        \includegraphics[width=.26\textwidth,height=0.17\textheight]{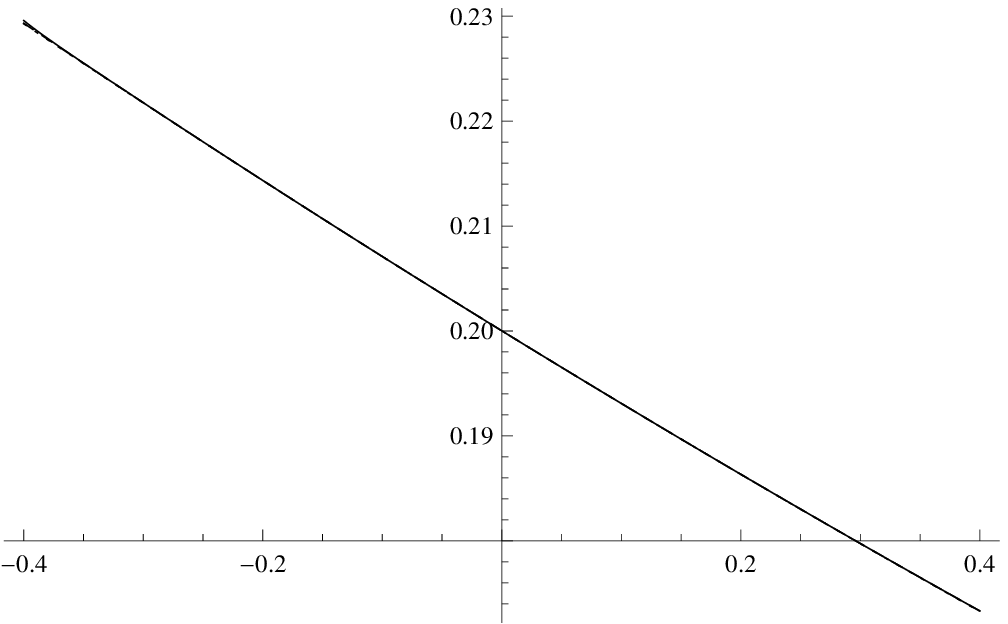} &
        \includegraphics[width=.26\textwidth,height=0.17\textheight]{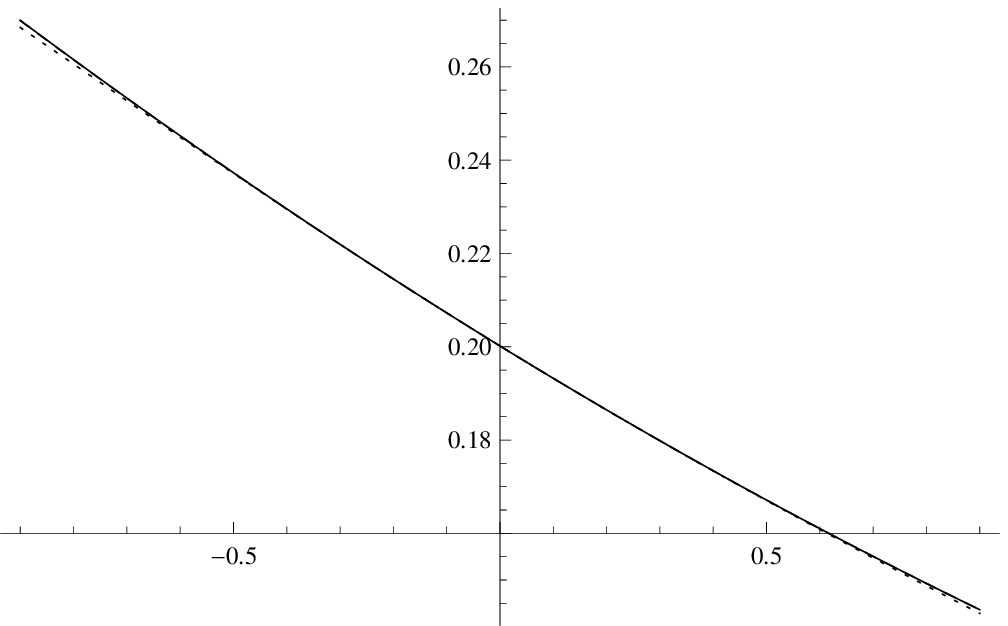}\\ \hline
        $t=5.0$ & $t=10.0$ \\
        \includegraphics[width=.26\textwidth,height=0.17\textheight]{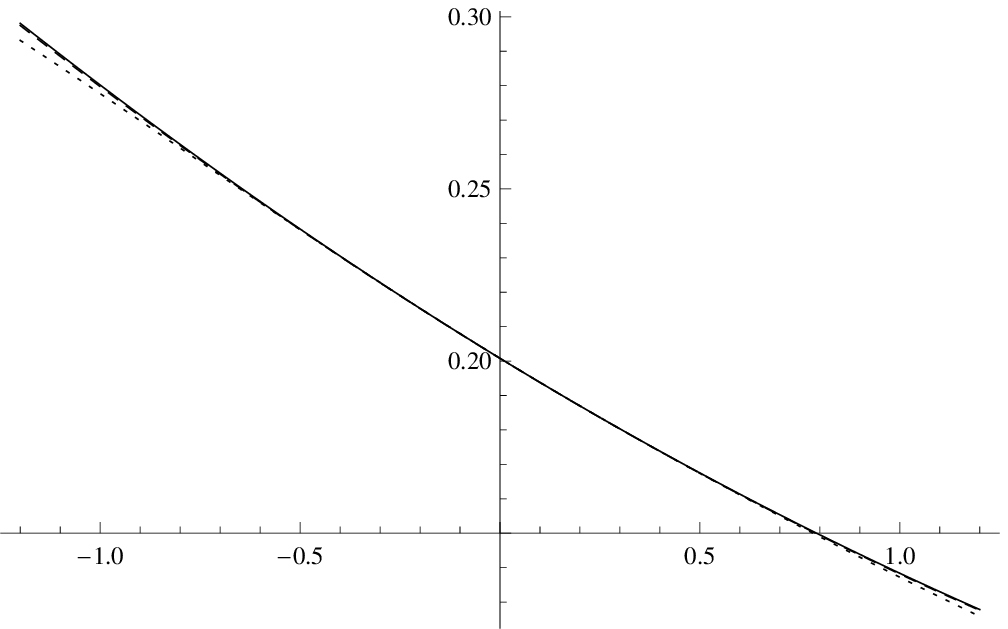} &
        \includegraphics[width=.26\textwidth,height=0.17\textheight]{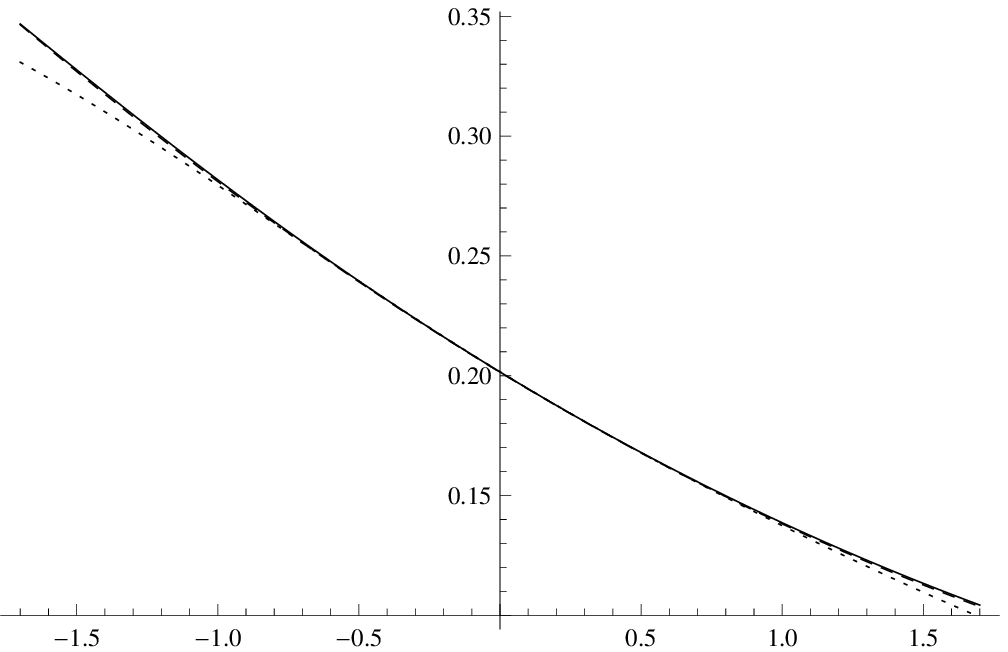}\\ \hline
    \end{tabular}
    \begin{tabular}{ c }
    \includegraphics[width=0.41\textwidth,height=0.40\textheight]{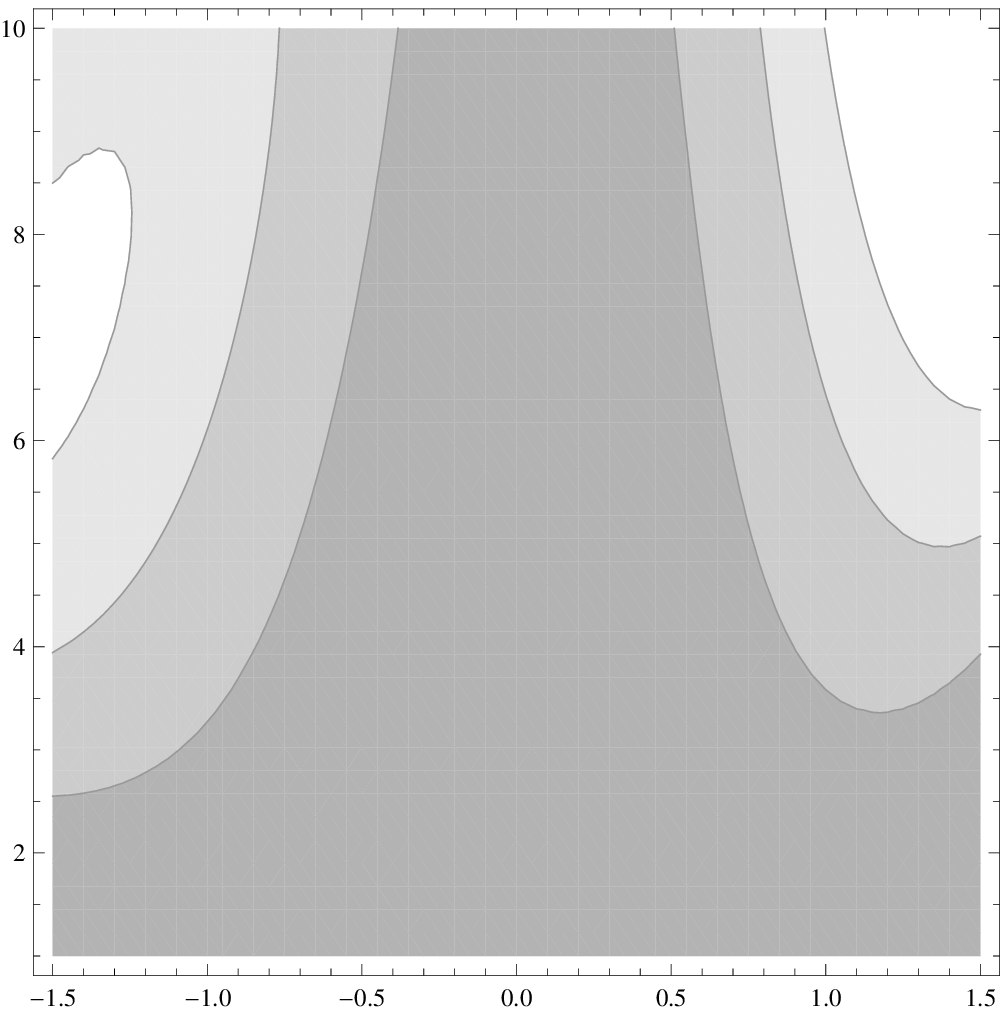}
    \end{tabular}
  \caption{
    LEFT:  Implied volatility in the CEV model \eqref{eq:model.CEV} is plotted as a function of $\log$-moneyness $(k-x)$ for four different maturities $t$.  The solid line corresponds to the implied volatility $\sig$ obtained by computing the exact price $u$ using \eqref{eq:u.CEV} and then by solving \eqref{eq:imp.vol.def} numerically.  The dashed line (which is nearly indistinguishable from the solid line) corresponds to our third order implied volatility approximation $\bar{\sig}_3$, which we compute by summing the terms in \eqref{eq:sig.CEV}.  The dotted line corresponds to the implied volatility expansion $\sig^\textrm{HW}$ of \cite{hagan-woodward}, which is computed using \eqref{eq:sig.hw}.
    RIGHT: We plot the absolute value of the relative error $|\bar{\sig}_3-\sig|/\sig$ of our third order implied volatility approximation as a function of $\log$-moneyness $(k-x)$ and maturity $t$.  The horizontal axis represents $\log$-moneyness $(k-x)$ and the vertical axis represents maturity $t$.  Ranging from darkest to lightest, the regions above represent relative errors of $<0.3\%$, $0.3\%$ to $0.6\%$, $0.6\%$ to $0.9\%$ and $>0.9\%$.
    We use the following parameters: $\beta=0.3$, $\del = 0.2$, $x = 0.0$.
    }
  \label{fig:cev}
\end{sidewaysfigure}

\begin{sidewaysfigure}
    \centering
    \begin{tabular}{ | c | c | }
        \hline
        $t=0.1$ & $t=1.0$ \\
        \includegraphics[width=.26\textwidth,height=0.17\textheight]{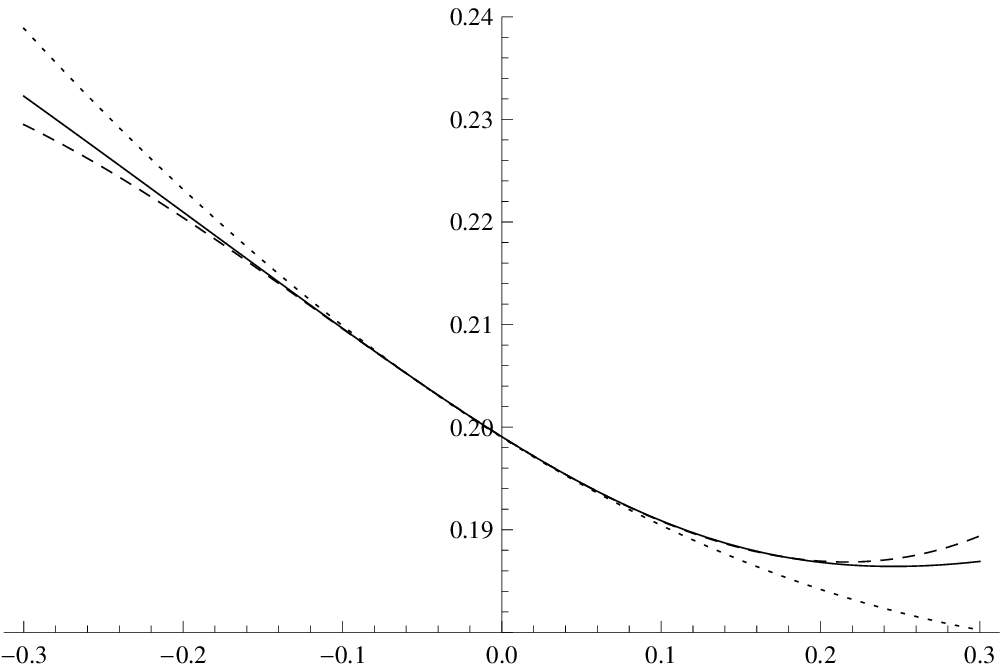} &
        \includegraphics[width=.26\textwidth,height=0.17\textheight]{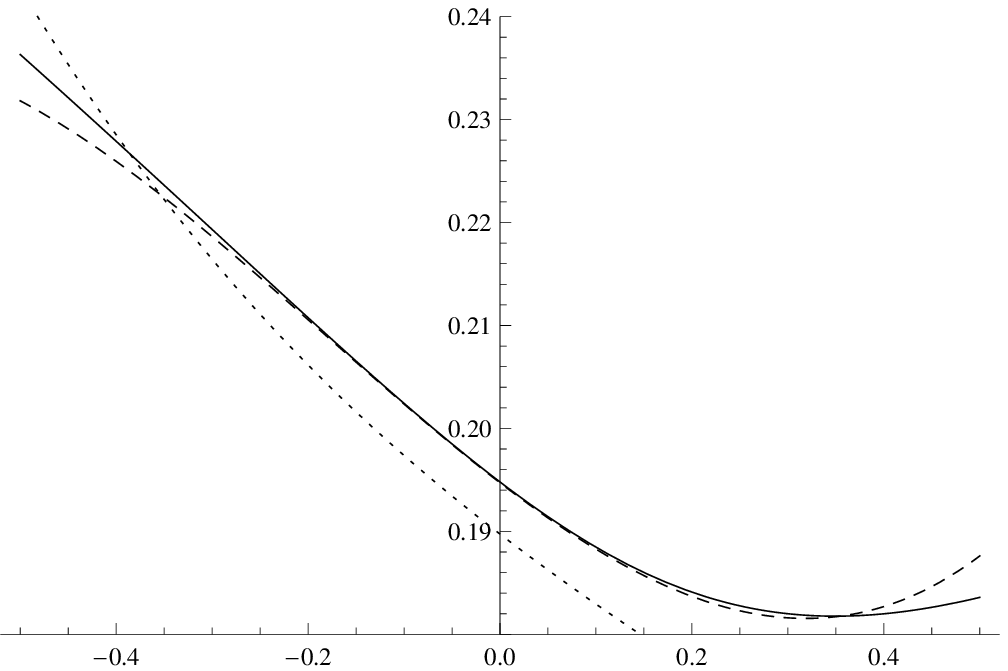}\\ \hline
        $t=5.0$ & $t=10.0$ \\
        \includegraphics[width=.26\textwidth,height=0.17\textheight]{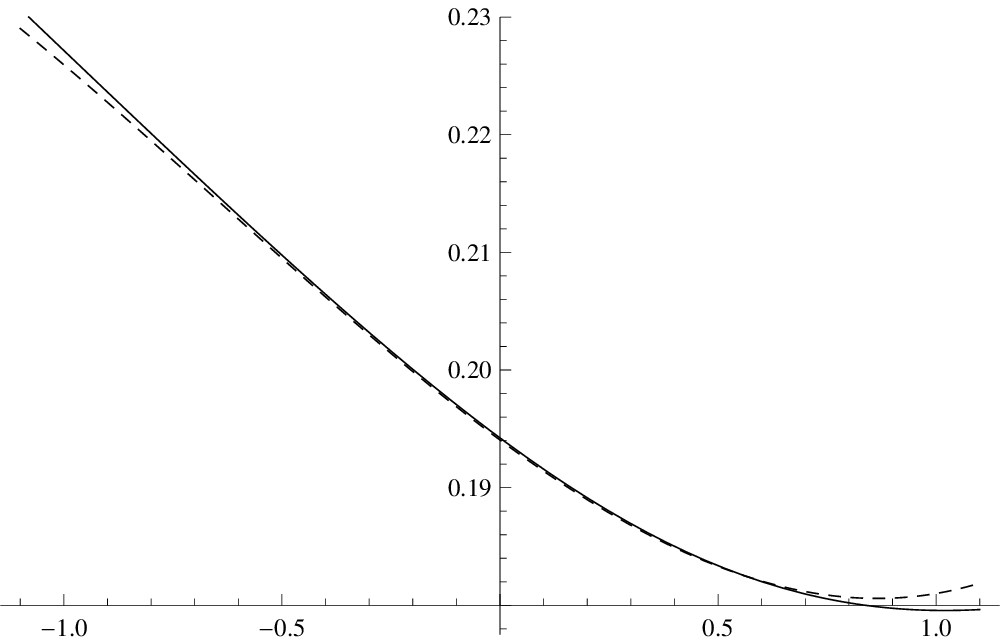} &
        \includegraphics[width=.26\textwidth,height=0.17\textheight]{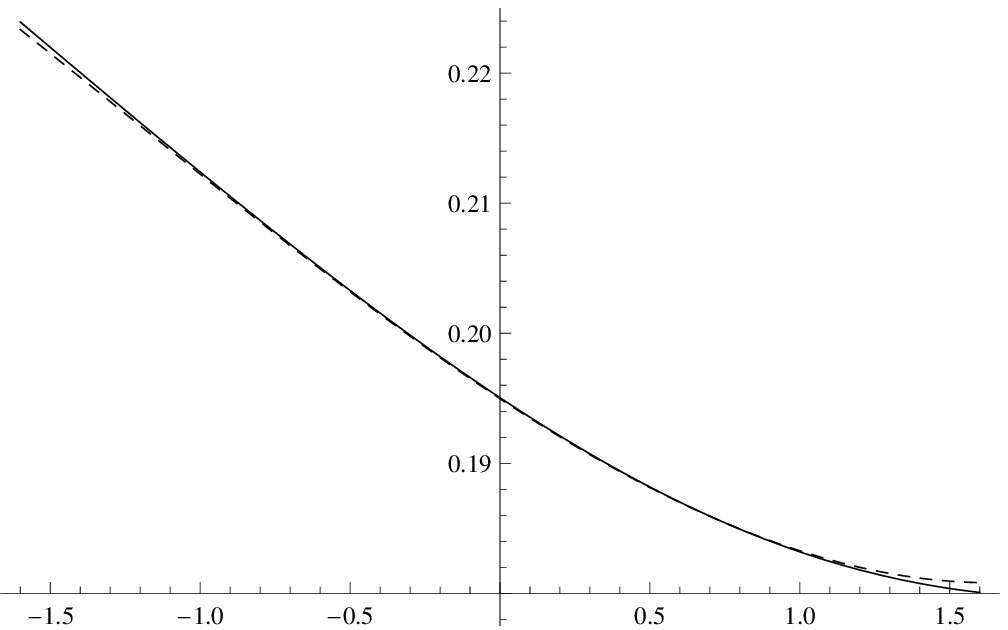}\\ \hline
    \end{tabular}
    \begin{tabular}{ c }
    \includegraphics[width=0.41\textwidth,height=0.40\textheight]{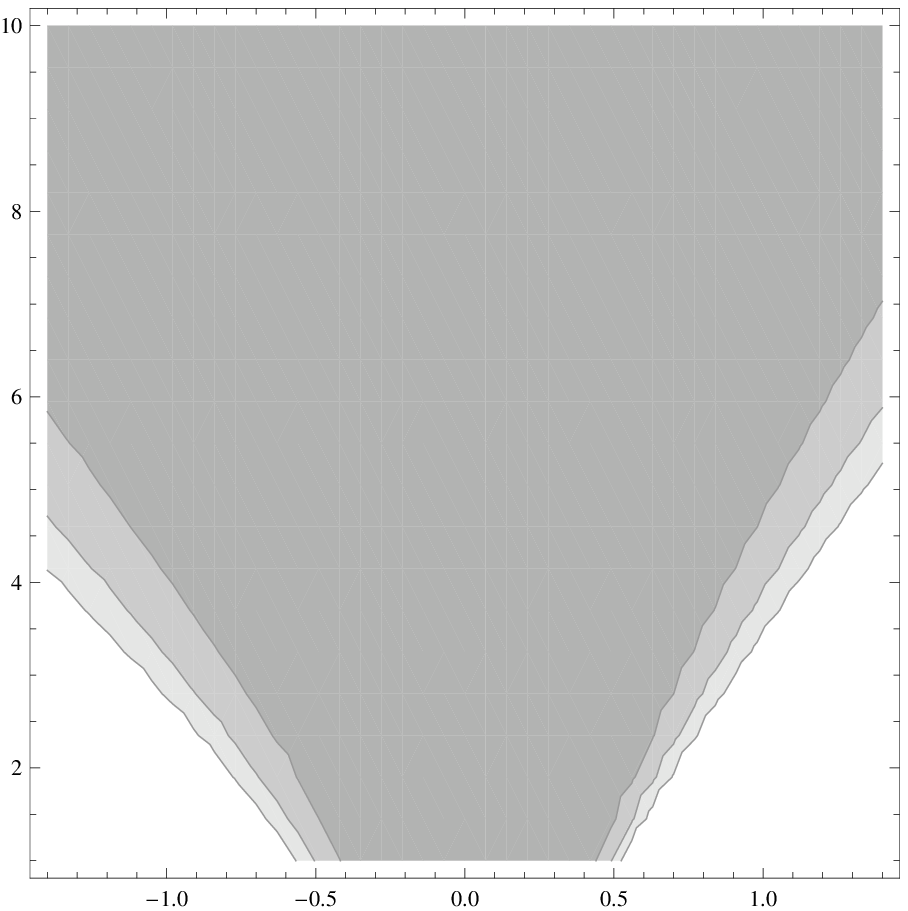}
    \end{tabular}
  \caption{
    LEFT:  Implied volatility the Heston model \eqref{eq:model.Heston} is plotted as a function of $\log$-moneyness $(k-x)$ for four different maturities $t$.  The solid line corresponds to the implied volatility $\sig$, obtained by computing the exact price $u$ using \eqref{eq:u.Heston} and then by solving \eqref{eq:imp.vol.def} numerically.  The dashed line corresponds to our third order implied volatility approximation $\bar{\sig}_3$, which we compute by summing the terms in \eqref{eq:sig.Heston} (note: $\sig_3$ does not appear in the text).  The dotted line (which only appears for the shortest two maturities) corresponds to the implied volatility expansion $\sig^\textrm{FJL}$ of \cite{forde-jacquier-lee}; it is computed using \eqref{eq:sig.fjl}.  Note that the dotted line does not appear in the plots for the two largest maturities.
    RIGHT: We plot the absolute value of the relative error $|\bar{\sig}_3-\sig|/\sig$ of our third order implied volatility approximation as a function of $\log$-moneyness $(k-x)$ and maturity $t$.  The horizontal axis represents $\log$-moneyness $(k-x)$ and the vertical axis represents maturity $t$.  Ranging from darkest to lightest, the regions above represent relative errors of $<1\%$, $1\%$ to $2\%$, $2\%$ to $3\%$ and $>3\%$.
    We use the parameters given in \cite{forde-jacquier-lee}:  $\kappa=1.15$, $\theta=0.04$, $\del=0.2$, $\rho=-0.40$ $x = 0.0$, $y = \log \theta$.
    }
  \label{fig:heston}
\end{sidewaysfigure}


\begin{sidewaysfigure}
    \centering
    \begin{tabular}{ | c | c | }
        \hline
        $t=0.1$ & $t=1.0$ \\
        \includegraphics[width=.26\textwidth,height=0.17\textheight]{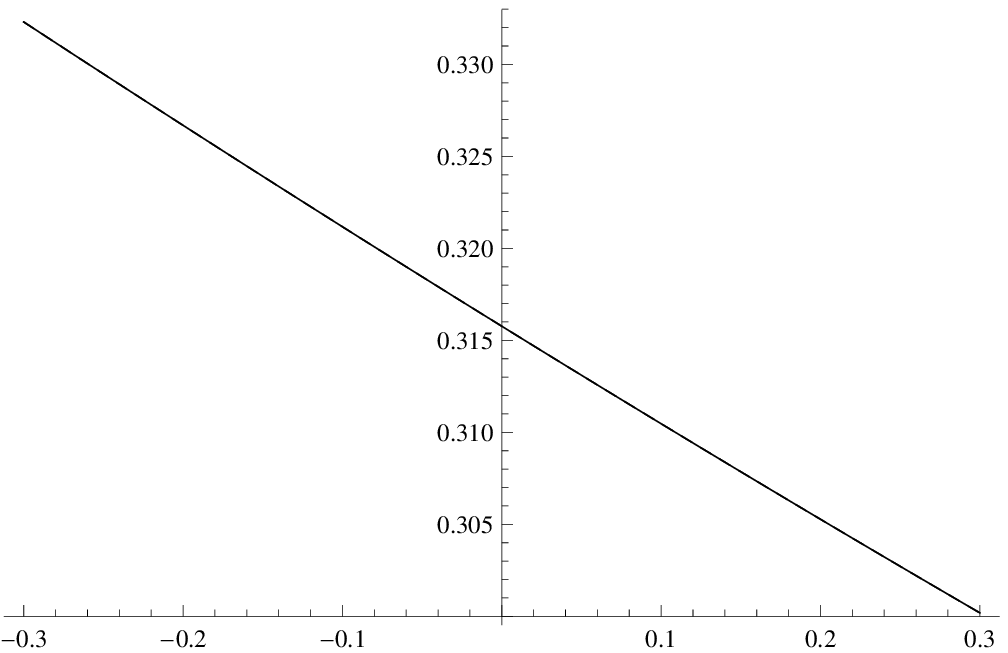} &
        \includegraphics[width=.26\textwidth,height=0.17\textheight]{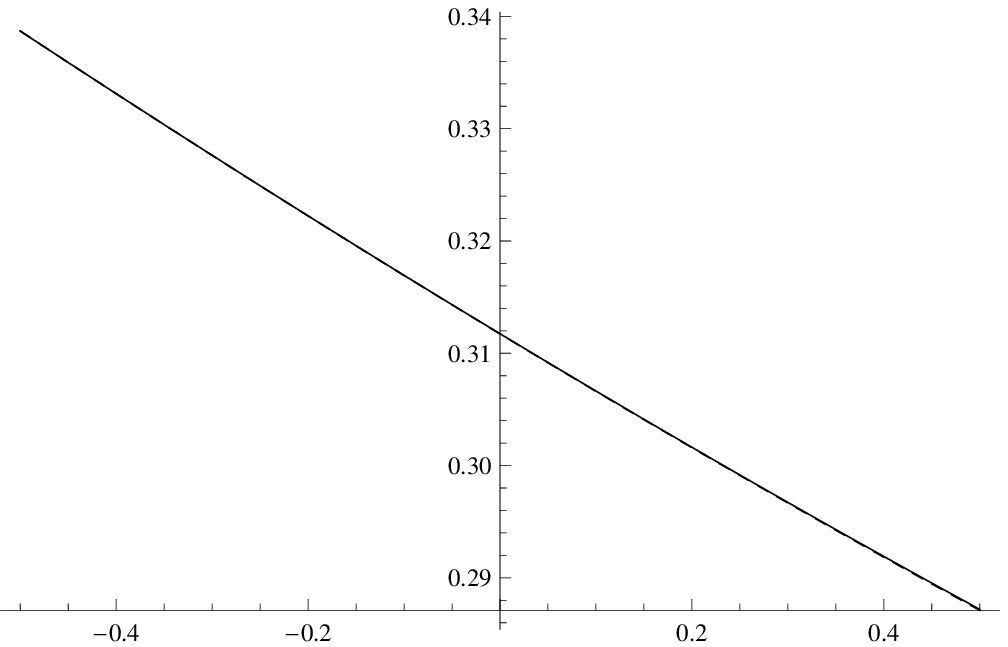}\\ \hline
        $t=3.0$ & $t=5.0$ \\
        \includegraphics[width=.26\textwidth,height=0.17\textheight]{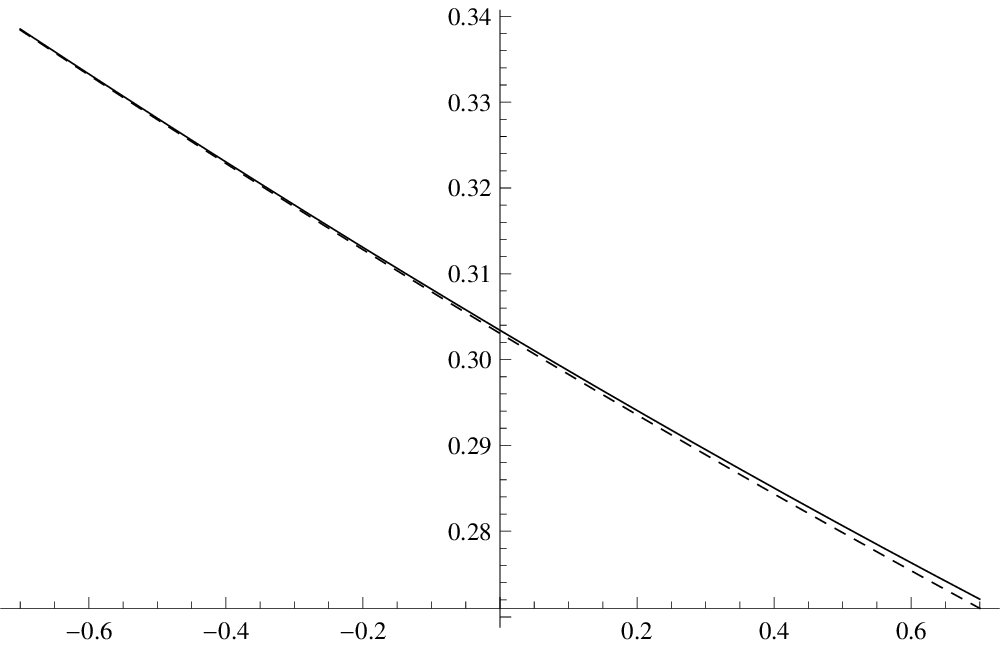} &
        \includegraphics[width=.26\textwidth,height=0.17\textheight]{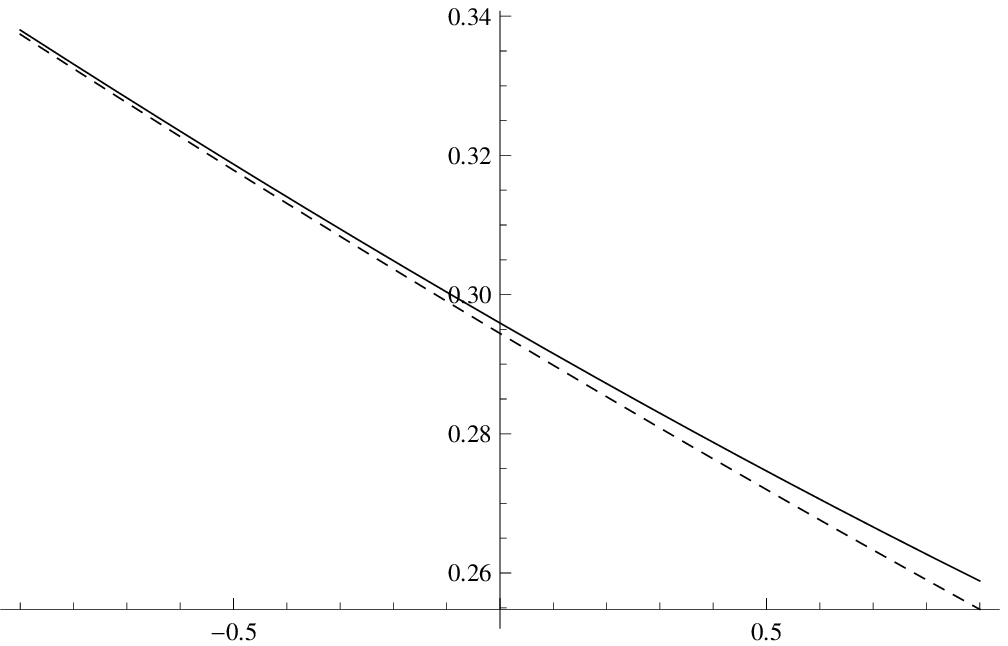}\\ \hline
    \end{tabular}
    \begin{tabular}{ c }
    \includegraphics[width=0.41\textwidth,height=0.40\textheight]{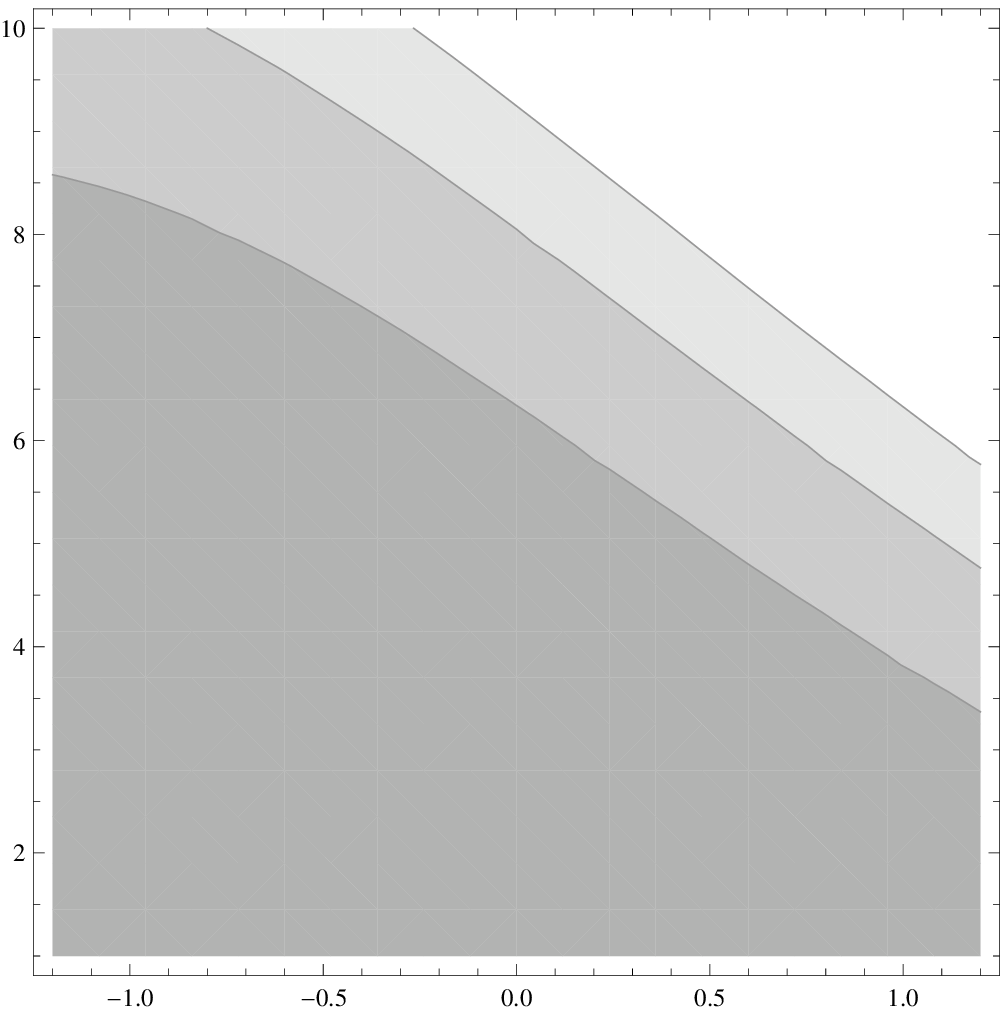}
    \end{tabular}
  \caption{
    LEFT:  Implied volatility in the 3/2 stochastic volatility model \eqref{eq:model.three-halves}
is plotted as a function of $\log$-moneyness $(k-x)$ for four different maturities $t$.  The solid
line corresponds to the implied volatility $\sig$, obtained by computing the exact
price $u$ using \eqref{eq:u.three-halves} and then by solving \eqref{eq:imp.vol.def} numerically.
The dashed line corresponds to our third order implied volatility approximation $\bar{\sig}_3$,
which we compute by summing the terms in \eqref{eq:sig.three-halves}.
    RIGHT: We plot the absolute value of the relative error $|\bar{\sig}_3-\sig|/\sig$ of our third order implied volatility approximation as a function of $\log$-moneyness $(k-x)$ and maturity $t$.  The horizontal axis represents $\log$-moneyness $(k-x)$ and the vertical axis represents maturity $t$.  Ranging from darkest to lightest, the regions above represent relative errors of $<1\%$, $1\%$ to $2\%$, $2\%$ to $3\%$ and $>3\%$.
    We use the following parameters: $\kappa=0.25$, $\theta=0.1$, $\del=0.8$, $\rho=-0.85$ $x = 0.0$, $y =
\log \theta$.
    }
  \label{fig:three-halves}
\end{sidewaysfigure}


\begin{sidewaysfigure}
    \centering
    \begin{tabular}{ | c | c | }
        \hline
        $t=0.1$ & $t=1.0$ \\
        \includegraphics[width=.26\textwidth,height=0.17\textheight]{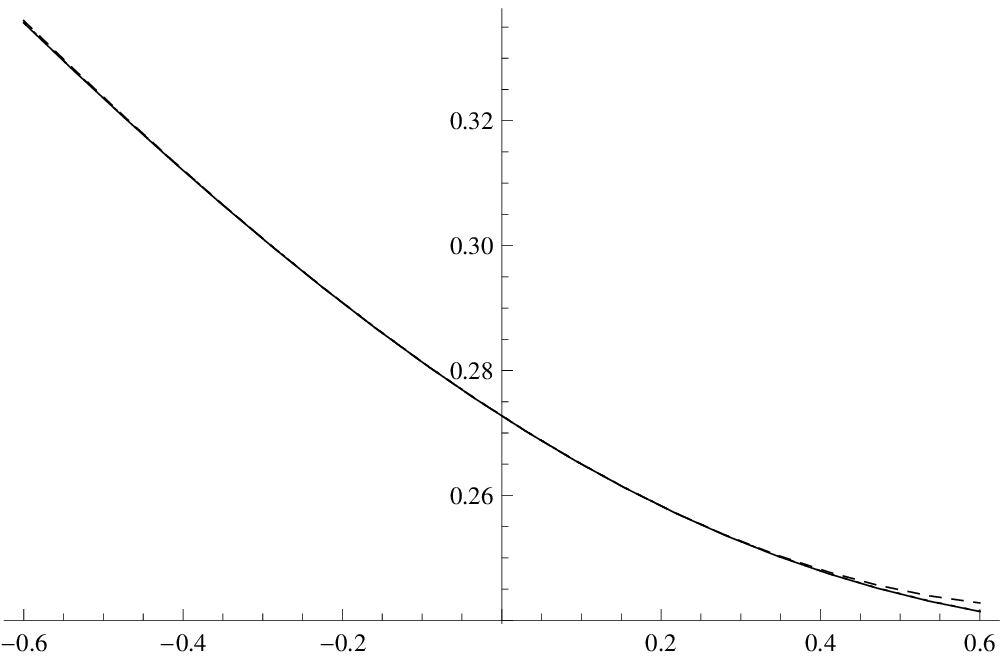} &
        \includegraphics[width=.26\textwidth,height=0.17\textheight]{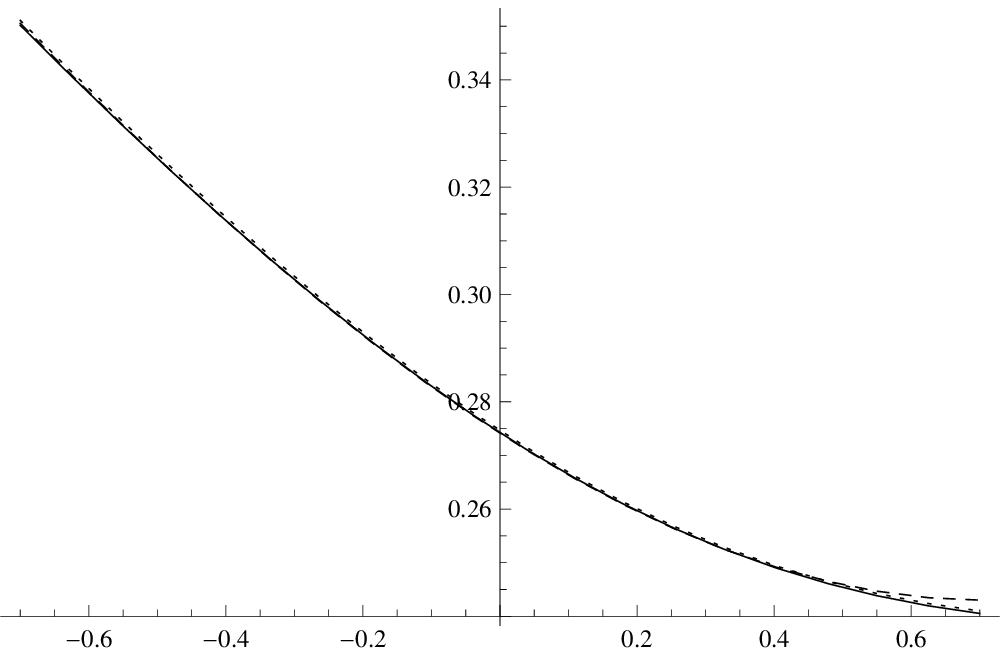}\\ \hline
        $t=5.0$ & $t=10.0$ \\
        \includegraphics[width=.26\textwidth,height=0.17\textheight]{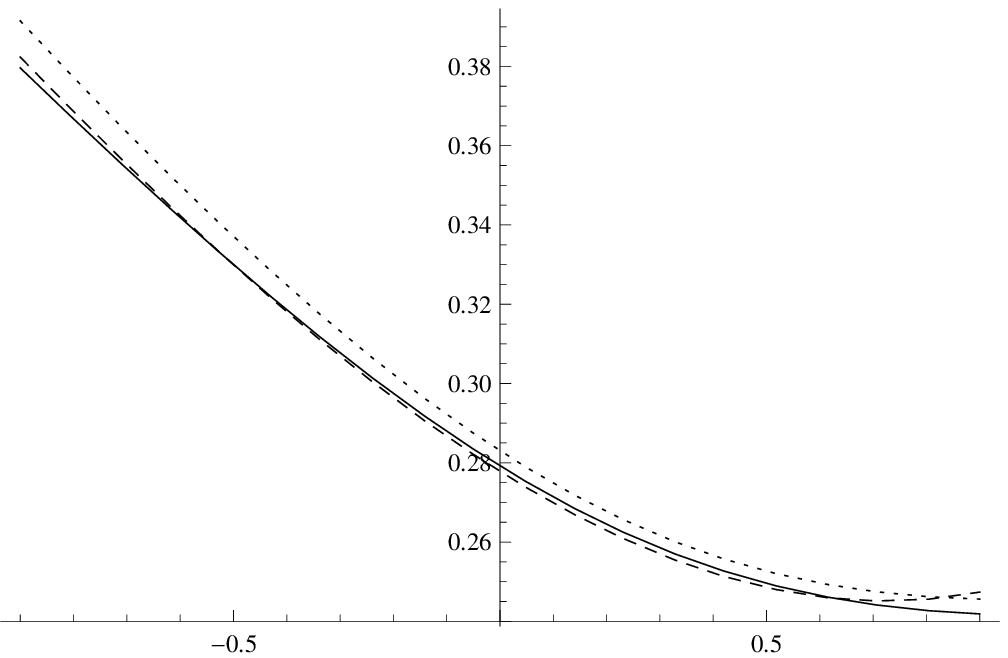} &
        \includegraphics[width=.26\textwidth,height=0.17\textheight]{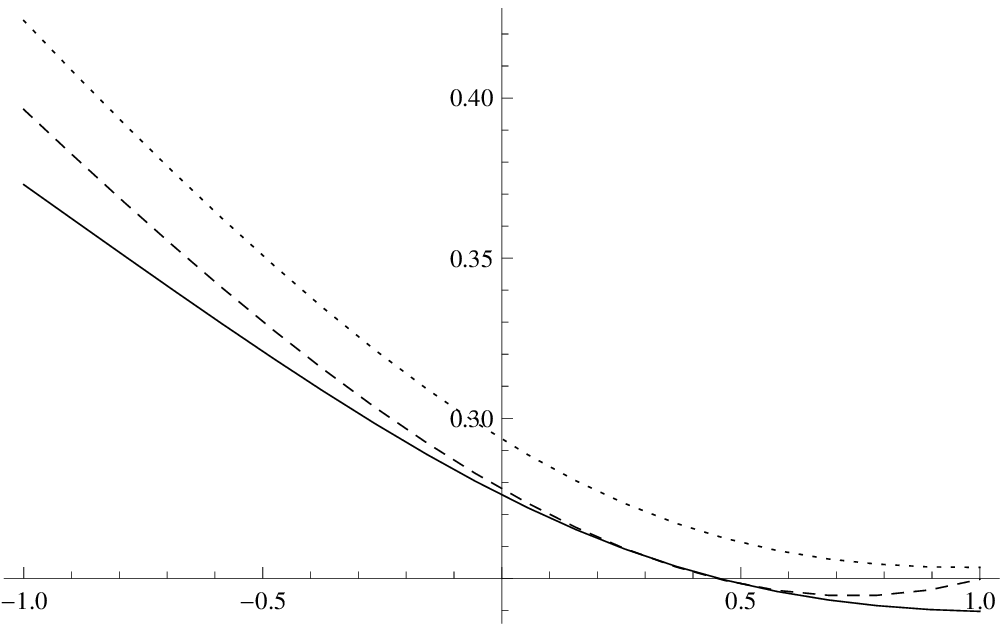}\\ \hline
    \end{tabular}
    \begin{tabular}{ c }
    \includegraphics[width=0.41\textwidth,height=0.40\textheight]{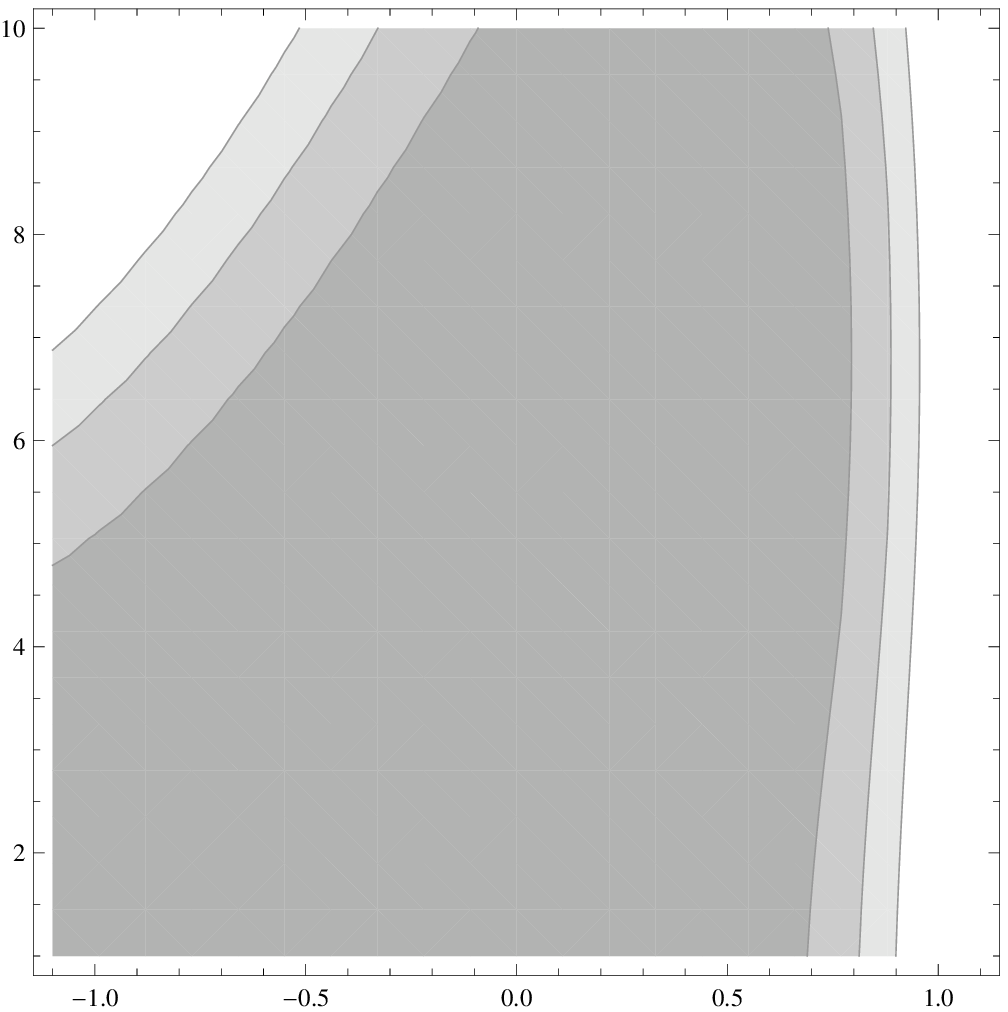}
    \end{tabular}
  \caption{
    LEFT: Implied volatility in the SABR model \eqref{eq:model.SABR} is plotted as a function of $\log$-moneyness $(k-x)$ for four different maturities $t$.  The solid line corresponds to the implied volatility $\sig$, obtained by computing the exact price $u$ using \eqref{eq:u.SABR} and then by solving \eqref{eq:imp.vol.def} numerically.  The dashed line corresponds to our third order implied volatility approximation $\bar{\sig}_3$, which we compute using \eqref{eq:sig.SABR}.  The dotted line corresponds to the implied volatility expansion $\sig^\textrm{HKLW}$ of \cite{sabr}, which is computed using \eqref{eq:sig.hklw}.  For the two shortest maturities, both implied volatility expansions $\bar{\sig}_3$ and $\sig^\textrm{HKLW}$ provide an excellent approximation of the true implied volatility $\sig$.  However, for the two longest maturities, it is clear that our third order expansion $\bar{\sig}_3$ provides a better approximation to the true implied volatility $\sig$ than does the implied volatility expansion $\sig^\textrm{HKLW}$ of \cite{sabr}.
    RIGHT: We plot the absolute value of the relative error $|\bar{\sig}_3-\sig|/\sig$ of our third order implied volatility approximation as a function of $\log$-moneyness $(k-x)$ and maturity $t$.  The horizontal axis represents $\log$-moneyness $(k-x)$ and the vertical axis represents maturity $t$.  Ranging from darkest to lightest, the regions above represent relative errors of $<1\%$, $1\%$ to $2\%$, $2\%$ to $3\%$ and $>3\%$.
    We use the following parameters: $\beta=0.4$, $\del=0.25$, $\rho=0.0$, $x=0.0$, $y=-1.3$.
    }
  \label{fig:sabr}
\end{sidewaysfigure}

\end{document}